\documentclass[sigconf]{acmart}

\AtBeginDocument{%
  \providecommand\BibTeX{{%
    Bib\TeX}}}

\copyrightyear{2025}
\acmYear{2025}
\setcopyright{rightsretained}
\acmConference[ICDCN 2025]{26th International Conference on Distributed Computing and Networking}{January 04--07, 2025}{Hyderabad, India}
\acmBooktitle{26th International Conference on Distributed Computing and Networking (ICDCN 2025), January 04--07, 2025, Hyderabad, India}
\acmPrice{}
\acmDOI{10.1145/3700838.3700851}
\acmISBN{979-8-4007-1062-9/25/01}

\usepackage[ruled,vlined,linesnumbered]{algorithm2e}
\let\oldnl\nl
\newcommand{\nonl}{\renewcommand{\nl}{\let\nl\oldnl}}
\def\BibTeX{{\rm B\kern-.05em{\sc i\kern-.025em b}\kern-.08em
    T\kern-.1667em\lower.7ex\hbox{E}\kern-.125emX}}

\usepackage[font=small,skip=3pt]{caption}
\usepackage{subcaption}
\usepackage{booktabs} 
\usepackage{bbm}
\usepackage{amsthm} 
\usepackage{amsmath} 
\newtheorem{theorem}{Theorem}

\usepackage[linesnumbered,ruled,vlined]{algorithm2e}
\SetKwInput{KwData}{Input}
\SetKwInput{KwResult}{Output}
\usepackage{algpseudocode}
\usepackage{mathtools}

\newtheorem{definition}{Definition}
\newtheorem{condition}{Condition}[definition]

\usepackage{enumitem}
\usepackage{setspace}
\DontPrintSemicolon

\newcommand{\ph}{\textbf{PreferenceHelper} \xspace}
\newcommand{\pc}{\textbf{PreferredCoalition}\xspace}
\newcommand{\pcdp}{\textbf{PCD}\xspace}
\newcommand{\pcgreedy}{\textbf{PCG}\xspace}

\newcommand{\our}{\textit{SMEVCA}}

\begin{document}
\title{SMEVCA: Stable Matching-based EV Charging Assignment in Subscription-Based Models}

\author{Arindam Khanda}
\email{akkcm@mst.edu}
\affiliation{
  \institution{Missouri University of Science and Technology, USA}
  \country{}
}
\author{Anurag Satpathy}
\email{anurag.satpathy@mst.edu}
\affiliation{%
  \institution{Missouri University of Science and Technology, USA}
  \country{}
}
\author{Anusha Vangala}
\email{anusha.vangala@mst.edu}
\affiliation{%
  \institution{Missouri University of Science and Technology, USA}
  \country{}
}
\author{Sajal K. Das}
\email{sdas@mst.edu}
\affiliation{%
  \institution{Missouri University of Science and Technology, USA}
  \country{}
}
\begin{abstract}
The rapid shift from internal combustion engine vehicles to battery-powered electric vehicles (EVs) presents considerable challenges, such as 
limited charging points (CPs), unpredictable wait times for charging, and difficulty in selecting appropriate CPs for EVs. To address these challenges, we propose a novel end-to-end framework, called \textbf{S}table \textbf{M}atching based \textbf{EV} \textbf{C}harging \textbf{A}ssignment (SMEVCA) that efficiently assigns charge-seeking EVs to CPs with the assistance of roadside units (RSUs). The proposed framework operates within a subscription-based model, ensuring that the subscribed EVs complete their charging within a predefined time limit enforced by a service level agreement (SLA). The framework SMEVCA employs a stable, fast, and efficient EV-CP assignment formulated as a one-to-many matching game with preferences. The matching process identifies the preferred coalition (a subset of EVs assigned to the CPs) using two strategies: (1) {\em Preferred Coalition Greedy} (PCG) that offers an efficient, locally optimal heuristic solution; and (2) {\em Preferred Coalition Dynamic} (PCD) that is more computation-intensive but delivers a globally optimal coalition. Extensive simulations reveal that PCG and PCD achieve a gain of 14.6\% and 20.8\% over random elimination for in-network charge transferred with only 3\% and 0.1\% EVs unserved within the RSUs vicinity.
\end{abstract}



\maketitle

\section{Introduction}\label{sec:Intro}
In recent years, the adoption of Electric Vehicles (EVs) has surged, driven by their use of onboard batteries that store electrical energy to power the vehicle's engine. Unlike traditional internal combustion engine (ICE) vehicles, EVs provide a cleaner energy alternative, reduce dependence on fossil fuels, and, most importantly, produce zero greenhouse gas emissions, paving the way for a more sustainable future \cite{yong2015review}. Many countries have made significant progress in this transition, viewing EVs as the future of transportation. For example, Canada and the UK have set policies to completely phase out ICE vehicles by $2040$, while China halted ICE production investments in $2019$ to promote EV sales. Additionally, about $10$ states in the U.S. have mandated the use of zero-emission vehicles, with California aiming for $5$ million EVs on the road by $2030$ \cite{shurrab2022stable}. Further, it is expected that by $2030$, the number of EVs on-road worldwide will exceed $250$ million \cite{shurrab2022stable}. 

With the anticipated rise in EVs on the road, the demand for charging these vehicles will significantly challenge the current infrastructure. Battery swapping — replacing discharged batteries with fully charged ones — was initially proposed to reduce the charging times. However, this approach required heavy forklift equipment for battery replacement, leading to wear and tear on the vehicle's battery compartment \cite{EV-1, EV-2}. Moreover, proprietary rights and ownership issues limited its commercial viability. As a result, plug-in EVs, which recharge by connecting to charging points (CPs), have gained popularity; the market share for fully electric and plug-in hybrid vehicles increased from 4\% in $2020$ to 18\% in $2023$ \cite{HEV, EVAdopt}. Despite this growth, the $2019$ Global EV Outlook and the Electric Vehicle Initiative (EVI) \cite{shurrab2022stable} reported that in $2018$, there were approximately $5.1$ million EVs on the road worldwide, but only $539,000$ public charging stations (CSs) \footnote{Note that a charging station can have multiple charging points}, 
with China accounting for $50$\% of these stations. In the U.S., a $2024$ report from the Department of Energy indicates there are currently $65,733$ public CSs with $179,173$ CPs and $3,641$ private CSs with $14,350$ CPs \cite{afdc_ev_charger_2018}. These statistics reveal a significant disparity between the growth rates of EVs and CSs, with EVs growing at a rate of $62.4$\% while CSs lag at $22.7$\%.
To put this in perspective, the EVI recommends a sustainable EV-to-CS ratio of $10:1$. However, in places like California, this ratio is far from being achieved, standing at $25:1$. Moreover, as more EVs hit the road, such ratio has worsened by $25$\% in recent years, indicating that the charging infrastructure is struggling to keep pace with the rising number of EVs. A major reason is higher installation costs in the CSs, consequently leading to sparse availability and immature deployments \cite{wang2016electric}. This motivates our work in this paper.

\vspace{-0.05in}
\subsection{Research Challenges}
The inadequate charging infrastructure presents significant research challenges as discussed below.

\vspace{2pt}
\noindent 1) {\em Range anxiety} is a critical issue that requires urgent attention \cite{liu2023reciprocal}. It describes the fear that EV owners experience, worrying that their vehicles may run out of charge before reaching the destination. This anxiety is particularly pronounced on highways and rural areas, where CSs are scarce. Running out of battery in such locations leaves the EV owners at risk of being stranded, potentially requiring a tow truck. The issue is further exacerbated as EV adoption expands into mixed utility vehicles until a sustainable highway and rural charging infrastructure is established. Range anxiety also affects urban areas, where crowded CSs result in longer wait times \cite{xiong2017optimal}.

\vspace{2pt}
\noindent 2) Unlike the ICE vehicles, where refueling time is predictable, limited availability of CSs and unpredictable nature of charging needs can lead to unbounded wait times, causing user dissatisfaction. 

\vspace{2pt}
\noindent 3) The charging ecosystem involves EV vendors installing CSs at rented sites, such as \texttt{Walmart}, \texttt{Costco}, and \texttt{Target}, or at standalone CSs. This setup incurs one-time capital expenditures and ongoing operational expenditures, with electricity costs being the primary contributor to the latter \cite{afdc_ev_charger_2018}. Therefore, making appropriate charging decisions is crucial to ensuring the economic sustainability of the charging infrastructure for EV vendors.

\vspace{2pt}
\noindent 4) Given the stochastic nature of wait times and charging preferences (e.g., minimal detour, fast charging options, and competitive pricing), finding an appropriate CP for an EV is a complex decision to be made in real-time, presenting a significant challenge \cite{shurrab2022stable}.

\vspace{-0.1in}
\subsection{Our Contributions}
To address the above challenges, we propose a novel framework called Stable Matching based EV Charging Assignment ({\our}). The major contributions are as follows.

\vspace{-0.05in}
\begin{itemize}
    \item We propose an \textit{end-to-end CP assignment framework}, wherein charge-seeking EVs submit their charging requests to a nearby roadside unit (RSU) and are promptly assigned to CPs.
    \item 
    We propose a novel {\em service level agreement (SLA)-driven subscription model} that guarantees EV recharging within a predefined time based on the subscribed plan.
    \item 
    We model the assignment of EVs to CPs as a one-to-many matching game and devise a stable, fast, and effective solution. The matching procedure is based on selecting the preferred coalition at the CPs. Specifically, we propose two  preferred coalition selection strategies, called PCG and PCD, based on the greedy and dynamic programming algorithms. The former is an efficient, locally optimal heuristic, while the latter is computation-intensive but globally optimal. 
    \item 
    To evaluate the proposed {\our} framework, we compare the performance of the PCD and PCG algorithms with random coalition formation. Extensive simulations confirm the performance improvements of the proposed scheme over in-network charge transfer, SLA misses, and execution times. 
\end{itemize}

The remainder of the paper is organized as follows. Section \ref{sec:RelWork} summarizes the related literature on EV assignment to CPs or CSs.
Section \ref{sec:system_model} outlines the system model and assumptions while Section \ref{sec:problem_formulation} formulates the EV-CP assignment problem. The solution approach is detailed in Section \ref{sec:solution_approach}, followed by the simulation results in Section \ref{sec:results}. Section \ref{sec:discussion_and_limitations} discusses the features and limitations of {\our}, and conclusions are drawn in Section \ref{sec:cnls}.

\vspace{-0.1in}
\section{Related Work}\label{sec:RelWork}

Research on EV charging encompasses several critical challenges, including (1) infrastructure establishment and management, (2) scheduling of EVs to CPs, and (3) pricing and subscription models for payment management. We focus on developing an effective scheduling algorithm for assigning EVs to CPs, and we primarily review the literature with this objective in mind.

The allocation of EVs to CSs involves optimizing charge usage while ensuring timely and cost-effective charging with minimal resource expenditure. Algafri \textit{et al.} \cite{ALGAFRI2024} addressed this by combining the Analytic Hierarchy Process (AHP) with goal programming to minimize costs, reduce battery degradation, and lower time overheads for EV owners. However, their approach does not account for the complexities of extended trips requiring visits to multiple CSs, which Sassi and Oulamara address \cite{rw-2017-1}. They proposed scheduling EVs to visit multiple CSs to maximize travel distance while minimizing costs, but this model overlooks real-time constraints related to CS availability.

Similarly, in car-sharing services requiring long-distance travel, $0-1$ mixed integer linear programming (0-1 MILP) can assign cars to CSs based on reservation time slots, as suggested in \cite{rw-2017-2}. However, this approach is constrained to scenarios with predetermined schedules and does not accommodate unpredictable demand or on-demand car-sharing allocations. Conversely, the EV assignment method in \cite{rw-2024} uses crowd detection and generalized bender decomposition (GBD) for optimized assignments, considering M/M/D queueing at fast CSs and a distributed biased min consensus algorithm for navigation. While this method offers advancements, it may introduce implementation complexities and high computational overhead. On the other hand, Elghitani and El-Saadany \cite{rw-2021-El} proposed a strategy to reduce the total service time across the EVs using Lyapunov optimization, dynamically reassigning EVs based on travel delays. However, this approach might lead to user dissatisfaction due to frequent reassignments, potentially disrupting travel plans and increasing anxiety. 

Though the existing models primarily focus on reducing charging costs for EV owners, they frequently neglect the financial implications for EV vendors. The model presented in \cite{rw-2019} is one of the first to address vendor profit maximization by assigning EVs to CSs based on available solar energy, aiming to serve the maximum number of vehicles while minimizing operational costs (OPEX). However, this approach introduces a drawback by charging users post-admission based on their remaining valuation, which could result in unfair pricing and diminished user satisfaction.   

\noindent \textbf{Research Gaps: }
The reviewed literature indicates that most studies on CP assignments \cite{rw-2024, rw-2021-El} are computationally intensive and time-consuming. In contrast, \cite{rw-2017-1} overlooks real-time constraints related to CS availability, waiting time limits, and charging type. Additionally, many studies neglect the financial impact on EV vendors, which is crucial for the sustainability of the charging infrastructure.

{\our} is designed to tackle these challenges by implementing a subscription-based SLA regulating waiting time limits to ensure appropriate CP assignments. A scalable, time-efficient, and matching-based procedure maximizes EV assignments in-network, effectively reducing economic overhead. 

\vspace{-0.05in}
\vspace{-0.05in}
\section{System Model and Assumptions}\label{sec:system_model}
A Roadside Unit (RSU) governs a circular region encompassing multiple CSs, as illustrated in Fig. \ref{fig:system_model}. Subscribed EVs enter the RSU’s coverage area within a given time frame, each requiring a charge. These EVs are assumed to be part of a vendor-specific plan tied to the established charging infrastructure in the region. The vendor usually offers two charging options, represented by a variable $\theta \in \{0,1\}$: (1) Level 3, or fast charging, denoted as $\theta = 1$, and (2) Level 2, or regular charging, denoted as $\theta = 0$. Fast charging (Level 3) is faster but comes at a higher cost compared to regular charging (Level 2) \cite{afdc_ev_charger_2018}. The vendor-based subscription framework aims to provide a seamless charging experience for EV owners by bounding the waiting time and eliminating the need for payment at each recharge. In this work, we do not delve into the pricing specifics of the subscription plan; instead, we focus on efficiently assigning CPs to EVs based on their subscribed plans. 

\begin{figure}[htb]
    \centering
    \includegraphics[width=0.99\columnwidth]{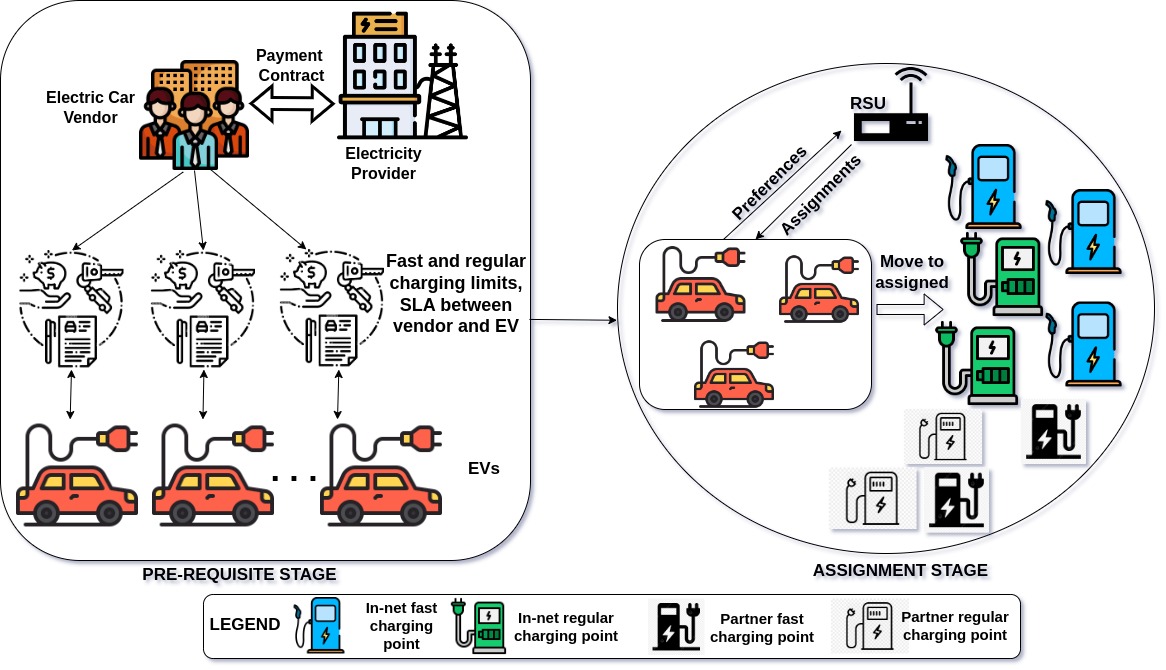}
    \caption{The Overall Architecture of {\our}.}
    \label{fig:system_model}
\vspace{-0.1in}
\end{figure}
\vspace{-0.15in}
\subsection{Subscription Model}\label{sec:subscription_model}
A subscription is a contractual agreement wherein an EV owner (subscriber) pays a recurring fee to receive charging services from a vendor over a specified subscription period. For a set of subscribed EVs represented as $\mathcal{S} = \{s_i\,|\,1 \leq i \leq |\mathcal{S}|\}$, the subscription plan has the following features.
\begin{itemize}[leftmargin=*]
    \item \textbf{Fast Charging Limit}: Each subscribed EV $s_i \in \mathcal{S}$ has an available fast charging capacity $f_i$ (in \textit{kWh}) for the subscription period. This limit applies within the vendor's network and extends to partner networks (see Section~\ref{subsec:charging_infrastructure}).
    \item \textbf{Regular Charging Limit}: A subscriber $s_i$ is entitled to unlimited regular charging throughout the subscription period.
    \item \textbf{Service Level Agreement (SLA)}:
    The agreement bounds the maximum wait time per charging session for subscriber $s_i$ at any assigned CP, irrespective of its association. This maximum time, denoted as $\gamma_i$, includes the wait time before the actual charging begins due to any existing queue at the assigned CP.
\end{itemize}
\vspace{-0.1in}
\subsection{Charging Infrastructure and Cost Model}
\label{subsec:charging_infrastructure}
Offering charging services involves two primary types of expenditures from the vendor's perspective: (1) Capital Expenditure (CAPEX), which includes the costs of installing charging points and acquiring or leasing charging sites, and (2) Operational Expenditure (OPEX), covering electricity costs, network services, bank transaction fees, and more \cite{afdc_ev_charger_2018}. In the prevalent pre-pay model, EV owners pay for each recharge session individually, with prices varying based on location, time of day, and anticipated traffic at the charging points. In contrast, the subscription-based model explored in this work allows an EV owner to make a one-time payment to the car vendor (e.g., Tesla, Rivian) for a charging plan that covers a specific subscription period (typically days). 
 
The RSU assigns EVs to vendor-owned CPs on a need-to-need basis. 
Furthermore, vendors can establish partnerships with other vendors to extend service coverage. From the vendor's perspective, charging infrastructure is classified into two types:
\begin{itemize}[leftmargin=*]
    \item \textbf{In-Network CPs}: In {\our}, CPs (installed and maintained) by a vendor at owned or rented sites such as \texttt{Target}, \texttt{Costco}, and \texttt{Walmart} are considered as in-network CPs, and the set is denoted as $\mathcal{C}^{\theta, in}$. 
    
    \item \textbf{Partner Network CPs}: 
    During peak loads, a vendor may allow EVs to charge at CPs from other vendors, referred to as partner network CPs. 
    This network of partner CPs is denoted as $\mathcal{C}^{\theta, par}$ helps extend service coverage.
   
\end{itemize}

Additionally, we assume that within the RSU coverage, a vendor has four options for charging an EV $s_i$: (1) in-network fast, (2) in-network regular, (3) partner-network fast, and (4) partner-network regular. While fast charging is more expensive than regular charging, partner networks generally impose facilitation fees, making fast and regular charging more costly than in-network options. Therefore, it is reasonable to conclude that, for a vendor to boost its long-term revenue, it must service most EVs in-network.
\vspace{-0.05in}
\vspace{-0.05in}
\subsection{Assumptions}
{\our} operates under the following assumptions: All EVs have identical battery capacities, denoted as \(B_{full}\), and the vendor’s application/software on the EV knows details about its designated path and charge requirements. When charging is needed, the EV requests a nearby RSU. The RSU’s coverage radius is assumed to be $1.5$ \textit{miles}, with routes confined to this area \cite{Siemens2024}. 
This framework focuses exclusively on public CSs owned by vendors or installed at rented sites, excluding home or vehicle-to-vehicle (V2V) charging \cite{yassine2023match}. We assume a subscribed EV with an available fast charging quota prefers fast charging over regular charging. Charging requests are collected and processed in batches at the RSU, simplifying the model by not accounting for EVs dynamically entering or leaving the network during the processing period. Future work will address this dynamic aspect. The RSU manages charging requests from vendor applications and assigns EVs to suitable CPs, with subscribers adhering to the assignment recommendations.
\vspace{-0.1in}
\section{Problem Formulation}\label{sec:problem_formulation} 
An EV requiring charge requests information on available CPs from the nearby RSU, which responds with the precise locations of CPs. The EV then generates preferences based on the detour distance from its current route. Once the preferences are generated, the assignment algorithm, executing at the RSU, performs allocation considering factors such as (1) requested charge, (2) CP availability, (3) charging limit specified in the subscription, and (4) maximum wait time as per the SLA. Note that the RSU's assignment follows the vendor's interests and targets: (1) minimizing the charging cost of EVs and (2) SLA adherence. This can be achieved by serving maximum charge-seeking EVs in-network, as it is a cheaper option. 
While this approach is cost-effective,
assignments in-network could elevate the waiting times, potentially violating SLAs.

\vspace{-0.1in}
\subsection{SLA-Compliant EV-CP Assignment} \label{CPassign}
We consider a distributed model for preference generation wherein the subscribed EVs independently compute their preferences for CPs. 
Let us consider a subscribed EV $s_i$, with a residual battery capacity of $\mathcal{B}_i^0$ (\textit{kWh}), average velocity $v_i$ (\textit{miles/second}), and average mileage of $m_i$ (\textit{miles/ kWh}) places a charging request with the RSU. The estimated time duration ($\mathcal{T}^0_{i, j}$) and charge expended ($\mathcal{B}_{i, j}$) in traversing to a CP $c_j^{\theta, \eta} \in \mathcal{C}$ located at distance $d_{i, j}$ (\textit{miles}) can be derived as per Eqs. (\ref{eqn:eqn_1}) and (\ref{eqn:eqn_2}), where $c^{\theta, \eta}_{j} \in \mathcal{C}$ is the CP under consideration. Note that $\theta \in \{0, 1\}$ represents the type of charging, either fast or regular, while $\eta \in \{in, par\}$ denotes the subscribed network, either in-network or partner network.

\begin{minipage}{.4\linewidth}
\begin{equation} \label{eqn:eqn_1}
    \mathcal{T}^0_{i, j} = \frac{d_{i,j}}{v_i} [seconds]
\end{equation}
\end{minipage}%
\begin{minipage}{.4\linewidth}
\begin{equation}\label{eqn:eqn_2}
    \mathcal{B}_{i, j} = \frac{d_{i, j}} {m_i} [kWh]
\end{equation}
\end{minipage}

Assuming the charging rate at $c^{\theta, \eta}_{j}$ be $r_j$ (kWh/seconds), and EV $s_i$'s reception rate as $\hat{r_i}$ (kWh/seconds), the effective charge transfer rate is bounded by $\min(r_j, \hat{r_i})$.
Let $s_i$ desires to reach $\alpha_i (0 < \alpha_i \leq 1)$ fraction of the full charge $\mathcal{B}_{full}$, hence the actual charge ($\psi_{i, j}$) needed for $s_i$ at CP $c^{\theta, \eta}_{j}$ 

is derived as per Eq. (\ref{eqn:eqn_3}), where $\alpha_i$ is assumed to be predetermined by the EV owner.
\begin{equation}\label{eqn:eqn_3}
\vspace{-0.1in}
    \psi_{i, j} = \alpha_i \cdot \mathcal{B}_{full} - (\mathcal{B}_i^0 - \mathcal{B}_{i,j}) [kWh]
\end{equation}

Next, the amount of time expended ($\mathcal{T}_{i, j}$) by $s_i$ at $c^{\theta, \eta}_{j}$ to receive a recharge of $\psi_{i, j}$ is computed as Eq. (\ref{eqn:eq_4}).
\begin{equation} \label{eqn:eq_4}
    \mathcal{T}_{i, j} = \frac{\psi_{i, j}}{\min(r_j, \hat{r_i})} [seconds]
\end{equation}

A CP $c^{\theta, \eta}_{j}$ may not always be free for immediate charging. 
Assuming it becomes available after a duration of $\mathcal{T}^{free}_j$ from the current instance.
However, we consider that the SLA $\gamma_i$ is triggered once $s_i$ arrives at $c^{\theta, \eta}_{j}$ and encompasses the waiting time before the start of charging. If $s_i$ is to be charged with $\psi_{i, j}$, 
then the \textit{effective duration to SLA-compliant charging} ($\delta_{i, j}$) is derived as Eq. (\ref{eqn:eqn_5}).
\begin{align}\label{eqn:eqn_5}
\vspace{-0.1in}
\delta_{i, j}  & = \mathcal{T}_{i,j} + \gamma_i &\text{ if } \mathcal{T}^{free}_j = 0 \notag\\
&= \mathcal{T}_{i, j} + \gamma_i - (\mathcal{T}^{free}_j - \mathcal{T}^0_{i,j}) &\text{ if }  0 < \mathcal{T}^{free}_j \leq \mathcal{T}^0_{i, j} + \gamma_i \notag\\
& =  0 &\text{ Otherwise}
\end{align}

Additionally, we consider $\vartheta_{i, j} \in \{0,1\}$ to be an indicator variable that equals $1$ if subscriber $s_i \in \mathcal{S}$ is assigned to CP $c^{\theta, \eta}_{j} \in \mathcal{C}$, $0$ otherwise. 
Assuming $\tau_{c}$ is the current time instance, let $s_i$ complete its charging at $c^{\theta, \eta}_{j}$ at time $\tau_{i, j}$. The objective of {\our} can then be formally defined as follows\footnote{In the paper, time instants are captured by $\tau$, and time duration's (intervals) are denoted by $\mathcal{T}$.}:
\begin{align}
\vspace{-0.1in}
& \max \sum_{s_i \in \mathcal{S}} \sum_{c^{\theta, \eta}_{j}\, \in\, C^{\theta,in}}  \vartheta_{i, j} \cdot \psi_{i, j} \label{eqn:obj_function} 
\end{align}
such that
\begin{align}
 &\quad{} \sum_{c^{\theta, \eta}_{j} \,\in \, \mathcal{C}} \vartheta_{i,j} \leq 1  \quad \forall s_i \in \mathcal{S}\label{constraint_1} \\
 & \quad{} \sum_{s_i \in \mathcal{S}} \vartheta_{i, j} \leq q^{\theta, \eta}_{j} \quad \forall c^{\theta, \eta}_{j} \in \mathcal{C} \label{constraint_2} \\
 & \quad{} 
 \vartheta_{i, j} (\tau_{i,j} - (\tau_{c} + \mathcal{T}^0_{i, j})) \leq \delta_{i, j} \quad \forall s_i \in \mathcal{S}, \,\, \forall c^{\theta, \eta}_{j} \in \mathcal{C}\label{constraint_3}
\end{align}
Eq. (\ref{eqn:obj_function}) captures the overall objective of {\our} and is dedicated to maximizing the in-network charging. Eq. (\ref{constraint_1}) captures the assignment constraint, which implies that an EV can be assigned to at most one CP. On the other hand, Eq. (\ref{constraint_2}) bounds the number of EVs assigned at a CP $c^{\theta, \eta}_{j} \in \mathcal{C}$, wherein $q^{\theta, \eta}_{j}$ reflects its maximum capacity\footnote{The term Capacity/Quota/Queue Size are used interchangeably in this paper.}. 
Any EV-CP assignment within the in-network or partner network must be SLA-compliant, as captured in Eq. (\ref{constraint_3}).

\begin{theorem}
SLA-compliant EV-CP assignment is $\mathcal{NP}$-Hard.
\vspace{-0.05in}
\end{theorem}
\begin{proof}
The SLA-compliant EV-CP assignment problem can be reduced in polynomial time to the $0/1$ knapsack problem by setting the number of in-network charging points $\mathcal{C}^{\theta, in}_j$ to $1$, the CP's capacity $q^{\theta, in}_j$ to the number of subscribed EVs $|\mathcal{S}|$, and the maximum wait time for each subscribed EV to $\gamma_{max}$. If all EVs request charging simultaneously and the CP is available, the last assigned EV can wait up to $\gamma_{max}$ without breaching its SLA. Therefore, all but the previous EV assigned to the CP should complete their charging in $\gamma_{max}$ time.
In this scenario, the problem becomes a $0/1$ knapsack problem where each EV is an item, the recharge time is the item's weight, the charge requirement is the item's value, and $\gamma_{max}$ is the knapsack's capacity. Since the $0/1$ knapsack problem is proven to be $\mathcal{NP}$-Hard \cite{martello1999dynamic}, this specific case of the SLA-compliant charging problem is also $\mathcal{NP}$-Hard.
The general problem, involving multiple CPs with varying capacities and subscribers with different wait times $\gamma_i$, represents a more generalized and complex version of this case. Given the special case is $\mathcal{NP}$-Hard, we can conclude that the general variant of the problem is at least equally hard.
\end{proof}
\vspace{-0.1in}
\section{Solution Approach}\label{sec:solution_approach}
This section highlights the transformation of the SLA-compliant EV-CP assignment as a one-to-many matching game and introduces some crucial terminologies (Section \ref{sec:matching}). The solution to the matching game is elaborated in Section \ref{sec:working_of_matching_game}.
\vspace{-0.1in}
\subsection{EV to CP Assignment as a Matching Game}\label{sec:matching}
The assignment of CPs to EVs can be modeled as a matching game between sets $\mathcal{S} = \{s_i\,|\,1< i< |\mathcal{S}| \}$ and $\mathcal{C} =\{c^{\theta, \eta}_{j} \,|\,1 < j < |\mathcal{C}| \}$ \cite{assila2018many}. In this one-to-many matching game, each CP $c^{\theta,\eta}_{j}$ has a available quota $q^{\theta, \eta}_{j}$, initialized to its capacity. Although only one EV can be charged at any time, we consider $q^{\theta, \eta}_{j}$ scheduling multiple EVs assigned to $c^{\theta, \eta}_{j}$ to account for the varied charge requirements and non-overlapping charging duration. 
Formally, the matching game can be defined as follows.
\begin{definition} (\textbf{Scheduling Game}):
\label{def:def_matching_game}
Given the sets $\mathcal{S}$ and $\mathcal{C}$, a one-to-many matching $\mu$ is an assignment, such that $\mu: \mathcal{S} \cup \mathcal{C} \rightarrow 2^{\mathcal{S} \, \cup \, \mathcal{C}}$, satisfying the following conditions.
\begin{condition}\label{def1:condition_1}
$\forall \, s_i \in \mathcal{S}$, $|\mu(s_i)| \leq 1$, $\mu(s_i) \subseteq \mathcal{C}$,
\end{condition}
\vspace{-0.1in}
\begin{condition}\label{def1:condition_2}
$\forall \, c^{\theta, \eta}_{j} \in \mathcal{C}$, $|\mu(c^{\theta, \eta}_{j})| \leq q^{\theta, \eta}_{j}$, $\mu(c^{\theta, \eta}_{j}) \subseteq \mathcal{S}$,
\end{condition}
\vspace{-0.1in}
\begin{condition}\label{def1:condition_3}
$s_i \in \mu(c^{\theta, \eta}_{j})$ $\Longleftrightarrow$ $c^{\theta, \eta}_{j} \in \mu(s_i)$.
\end{condition}
\end{definition}

Note that Condition (\ref{def1:condition_1}) states that $s_i$ can be assigned to at most one CP. On the other hand, Condition (\ref{def1:condition_2}) states that $c^{\theta, \eta}_{j}$ can accommodate at most $q^{\theta, \eta}_{j}$ EVs. Finally, Condition (\ref{def1:condition_3}) enforces that $s_j$ is assigned to $c^{\theta, \eta}_{j}$ \textit{iff} they are matched to each other.

Note that in the context of matching theory, each agent $a \in \mathcal{S} \cup \mathcal{C}$ assigns preferences over agents in the opposite set \cite{echenique2004theory}.
These preference relations are \textit{binary}, \textit{anti-symmetric}, and \textit{complete} relations. For example, a subscribed EV $s_i \in \mathcal{S}$, can have a preference in the form of $c^{\theta, \eta}_{j} \succ_{s_i} c^{\theta, \eta}_{j'}$. This implies that $s_i$ prefers assignment with $c^{\theta, \eta}_{j}$ over $c^{\theta, \eta}_{j'}$. Moreover, such a preference relationship where an agent is indifferent between any two agents of the other set is called \textit{strict}. The accumulation of all such preferences of an agent $\{\succ_{s_i}\}_{i \in [1, |\mathcal{S}|]}$ is termed as its preference profile \cite{satpathy2022rematch}.

\begin{definition}(\textbf{Individual Rationality})\label{label:def_2:individual_rationality}
A matching $\mu$ is individually rational, \textit{iff} $\mu(s_i) \succ_{s_i} \emptyset$ and $\mu(c^{\theta, \eta}_{j}) = \Theta_j (\mu(c^{\theta, \eta}_{j}))$, where $\Theta_j(.)$ returns the most preferred coalition, i.e., subset of $\mathcal{S}$ as per the preference of $c^{\theta, \eta}_{j}$.
\end{definition}

The property of individual rationality is essential to maintain, as it ensures that the matching assignment is voluntary and that no agent in $a \in \mathcal{S} \cup \mathcal{C}$ prefers to be assigned to $\mu'(a) \subset \mu(a)$, i.e., $\mu'(a) \succ_{a} \mu(a)$. If such scenarios arise, agent $a$ will not prefer to be assigned to agents in $\mu(a) \setminus \mu'(a)$ \cite{liu2018dats}.

\begin{definition}(\textbf{Pairwise Block})\label{label:def_2:pairwise_block}
A matching $\mu$ is pairwise blocked by an EV-CP pair $(s_i, c^{\theta, \eta}_{j})$ \textit{iff} the following conditions are satisfied:
\begin{condition}\label{def_3:condition_1}
$s_i \notin \mu(c^{\theta, \eta}_{j})$ and $c^{\theta, \eta}_{j} \notin \mu(s_i)$    
\end{condition}
\vspace{-0.1in}
\begin{condition}\label{def_3:condition_2}
$c^{\theta, \eta}_{j} \succ_{s_i} \mu(s_i)$
\end{condition}
\vspace{-0.1in}
\begin{condition}\label{def_3:condition_3}
$s_i \in \Theta_j(\mu(c^{\theta, \eta}_{j})\, \cup \, \{s_i\})$ 
\end{condition}
\end{definition}
Condition (\ref{def_3:condition_1}) states that $s_i$ and $c^{\theta, \eta}_{j}$ are not assigned to each other. Conditions (\ref{def_3:condition_2}) and (\ref{def_3:condition_3}) enforce that $s_i$ and $c^{\theta, \eta}_{j}$ are not satisfied with their current assignments and have incentives to deviate and be matched to each other. 

\begin{definition} (\textbf{Stability}) \label{def:stability}
A matching $\mu$ is said to be stable \textit{iff} it is individually rational and is not pairwise blocked by any two agents. 
\end{definition}

\vspace{-0.1in}
\subsection{Working of the Matching Algorithm}\label{sec:working_of_matching_game}
The overall working of the matching algorithm encompasses two steps: (1) preference generation (refer to section \ref{sec:preference_generation}) and (2) assignment of EVs to CPs (refer to Section \ref{sec:matching_algorithm}). Next, we delve into the details of the preference generation of subscribed EVs \cite{zhao2016many}.
\vspace{-0.1in}
\subsubsection{Preferences of EVs}\label{sec:preference_generation}
\setlength{\textfloatsep}{0.1cm}
\setlength{\floatsep}{0.1cm}
\begin{algorithm}[htb]
    \caption{\textbf{Preference\_Generator\_EV}}
    \label{algo:pref_generator_EV}
    \small
    \DontPrintSemicolon
    \KwIn{$\mathcal{S}$, $\mathcal{C}$}
    \KwOut{$P(s_i), \forall s_i \in \mathcal{S}$}
    \textbf{Initialize}: $P(s_i) \gets \emptyset, \forall s_i \in \mathcal{S}$ \tcp*{Null Initialization}
    \For{$s_i \in \mathcal{S}$}{
        \tcc{\textbf{In-network:} More Preferred}
        \tcc{In-Network Fast Charging Points}

        $P^{1, in}_i \gets \ph (\mathcal{C}^{1, in}, s_i)$ \;
        Append $P^{1, in}_i$ to $P(s_i)$\;

        \tcc{In-Network Regular Charging Points}
        $P^{0, in}_i \gets \ph (\mathcal{C}^{0, in}, s_i)$\;
        Append $P^{0, in}_i$ to $P(s_i)$\;
        \tcc{\textbf{Partner network:} Less Preferred}
        \tcc{Partner Network Fast Charging Points}
        
        $P^{1, par}_i \gets \ph (\mathcal{C}^{1, par}, s_i)$\;
        Append $P^{1, par}_i$ to $P(s_i)$\;
        \tcc{Partner Network Regular Charging Points}
        $P^{0, par}_i \gets \ph (\mathcal{C}^{0, par}, s_i)$\;
        Append $P^{0, par}_i$ to $P(s_i)$\;
    }    
\end{algorithm}

Algorithm \ref{algo:pref_generator_EV} takes the set of $\mathcal{S}$ and $\mathcal{C}$ as input. It generates the preferences profiles $P(s_i), \forall s_i \in \mathcal{S}$ considering the set of feasible CPs $c^{\theta, \eta}_{j} \in \mathcal{C}$. 
 
For a subscribed vehicle $s_i$, the preferences over the CPs are ordered depending on ownership, i.e., whether in-network or partner network. Note that in-network CPs are given higher priority in this work. The rationale is to service maximum EVs in-network, thereby minimizing costs for the vendor, which is the primary goal of this paper.

Algorithm \ref{algo:pref_generator_EV} proceeds as follows. For every subscribed EV, the CPs are processed in the order: (1) In-network fast ($\mathcal{C}^{1, in}$), (2) In-network regular ($\mathcal{C}^{0, in}$), (3) Partner network fast ($\mathcal{C}^{1, par}$), and (4) Partner network regular ($\mathcal{C}^{0, par}$) (Lines 2-10 of Algorithm \ref{algo:pref_generator_EV}). 

For each category, individual CPs are ranked using Algorithm \ref{algo:PrefHelper}, wherein each CP $c_{j}^{\theta, \eta} \in \mathcal{C}^{\theta, \eta}$ is tested for its reachability given the current state of charge of $s_i$. If reachable, the procedure tests the available limit for fast or regular charging per the subscribed plan (Lines 1 - 5 of Algorithm \ref{algo:PrefHelper}).
The ordering within each sub-category is performed concerning the detour distances $\{d_{i, j}, d_{i, j'} \} \in \mathcal{D}_i$ for CPs $c_{j}^{\theta, \eta}$ and $c_{j'}^{\theta, \eta}$. Note that $\mathcal{D}_i$ is the vector containing detour distances of $s_i, \forall c_{j}^{\theta, \eta} \in \mathcal{C}^{\theta, \eta}$ within the RSU region. 

Therefore, for any two CPs in the same category, the preference relationship can be represented as:
\begin{equation*}
 c^{\theta, \eta}_{j} \succ_{s_i} c^{\theta, \eta}_{j'} \iff d_{i, j} < d_{i, j'}
\end{equation*}
\setlength{\textfloatsep}{0.1cm}
\setlength{\floatsep}{0.1cm}
\begin{algorithm}[htb]
    \caption{\textbf{\ph}}
    \label{algo:PrefHelper}
    \small
    \DontPrintSemicolon

    \KwIn{$\mathcal{C}^{\theta, \eta}$, $s_i$,  $\mathcal{D}_i$} 
 
    \KwOut{$P_i^{\theta, \eta}$}
     \textbf{Initialize}: $P_i^{\theta, \eta} \gets \emptyset$ \tcp*{Null Initialization}
    \For{$c^{\theta, \eta}_{j} \in \mathcal{C}^{\theta, \eta}$}{
        \If{$(\mathcal{B}_i^0 - \mathcal{B}_{i,j}) > 0$}{
            \If{$\theta = 0$ OR $f_i > \psi_{i,j}$}
            {
                $P_i^{\theta, \eta} \gets P_i^{\theta, \eta} \bigcup \{ c^{\theta, \eta}_{j} \} $
            } 
        }      
    }
    $P(s_i^{\theta, \eta}) \gets$ \textbf{Increasing\_Order\_Sort} ($P_i^{\theta, \eta}, \mathcal{D}_i)$ \\
    return $P_i^{\theta, \eta}$\;
\end{algorithm}

\subsubsection{Matching Algorithm}\label{sec:matching_algorithm}
The overall assignment procedure executes as per Algorithm \ref{algo:Matching}. It takes as input  $\mathcal{C}$, $\mathcal{S}$, and the preferences of agents $P(s_i)$, $\forall s_i \in \mathcal{S}$. The initialization encompasses (1) setting the initial matching of agents in $\mathcal{S} \cup \mathcal{C}$ to $\emptyset$, (2) setting the waiting ($\mathbb{W}^{\theta, \eta}_j$) and current proposer ($\mathbb{P}^{\theta, \eta}_j$) lists for each $c^{\theta, \eta}_{j} \in \mathcal{C}$ to $\emptyset$. The proposer and waiting lists are updated at every iteration of the matching. 
The former tracks the proposals received at a specific round, whereas the latter records the accepted proposals (the preferred coalition) till this round. 

Algorithm \ref{algo:Matching} proceeds in two phases. In the first phase, i.e., \textit{proposal phase}, every unassigned $s_i \in \mathcal{S}'$, with a non-empty preference, identifies its most preferred charging option (say $c^{\theta, \eta}_{j}$) using \textbf{Extract(.)}, and sends a proposal to the same (Algorithm \ref{algo:Matching} Lines 5-9). After sending the proposal, the same is evicted from the preference ($P(s_i)$) and is added to the current proposal list ($\mathbb{P}^{\theta, \eta}_j$) of $c^{\theta, \eta}_{j}$.
\begin{algorithm}[]
    \caption{\textbf{Matching}}
    \label{algo:Matching}
    \DontPrintSemicolon

    \small
    \KwIn{$\mathcal{C}, \mathcal{S}, P(s_i), \, \forall \, s_i \in \mathcal{S}$}
    \KwOut{$\mu: \mathcal{S} \cup \mathcal{C} \rightarrow 2^{\mathcal{S} \, \cup \, \mathcal{C}}$}
 
    \textbf{Initialize}: $\mu(s_i) \gets \emptyset$, $\forall \, S_{i} \in \mathcal{S}$, $\mu(c^{\theta, \eta}_{j}) \gets \emptyset$, $\forall \, c^{\theta, \eta}_{j} \in \mathcal{C}$ \newline
    Current Proposer List: $\mathbb{P}^{\theta, \eta}_j \gets \emptyset, \forall c^{\theta, \eta}_{j} \in \mathcal{C}$ \newline
    Waiting List: $\mathbb{W}^{\theta, \eta}_j \gets \emptyset, \forall c^{\theta, \eta}_{j} \in \mathcal{C}$\;
    \While{$True$}
    {
        $\mathcal{S}' \gets \{ s_i \,|\, \mu(s_i) = \emptyset \,\, \&\& \,\, P(s_i) \neq \emptyset$\}\;
        \If{$\mathcal{S}' \neq \emptyset$}
        {
            \For{each $s_i \in \mathcal{S}'$}
            {
                $c^{\theta, \eta}_{j} \gets$ \textbf{Extract} ($P(s_i)$)
                \tcp{Most Preferred Choice}
                Send Proposal to $c^{\theta, \eta}_{j}$\;
                Remove $c^{\theta, \eta}_{j}$ from $P(s_i)$\;
                $\mathbb{P}^{\theta, \eta}_j \gets \mathbb{P}^{\theta, \eta}_j \cup \{s_i\}$\;
            }
            \For{each $c^{\theta, \eta}_{j} \in \mathcal{C}$ such that $\mathbb{P}^{\theta, \eta}_j \neq \emptyset$}
            {
                $\mu(c^{\theta, \eta}_{j}) \gets \pc(\mathbb{W}^{\theta, \eta}_j, \mathbb{P}^{\theta, \eta}_j)$ \tcp{Select the preferred coalition from ($\mathbb{W}^{\theta, \eta}_j \cup \mathbb{P}^{\theta, \eta}_j$)}
                $\mathbb{W}^{\theta, \eta}_j \gets \mu(c^{\theta, \eta}_{j})$\;
                $\mathbb{P}^{\theta, \eta}_j \gets \emptyset$ 
                    \tcp{Other applicants are rejected}
            } \vspace{-0.04in}
        }
    \Else
    {
        \For{$c^{\theta, \eta}_{j} \in \mathcal{C}$}
        {
            $\mu(c^{\theta, \eta}_{j}) \gets 
            \mathbb{W}^{\theta, \eta}_j$\;
        }
        \For{$s_i \in \mathcal{S}$}
        {
            $\mu(s_i) \gets c^{\theta, \eta}_{j}$ 
            \text{such that} \,\, $s_i \in \mathbb{W}^{\theta, \eta}_j$\;
        }
    return $\mu$\;
    } \vspace{-0.04in}
}
\end{algorithm}

Once the EVs send their proposals, the \textit{accept/reject} phase of the algorithm is triggered. In this phase, the procedure first identifies the CPs that have received at least one proposal. For such $c^{\theta, \eta}_{j}$, the preferred coalition is computed using Algorithm \ref{algo:PreferredCoalition}, which computes the most feasible coalition from the set $\mathbb{L}^{\theta, \eta}_j$ using two different procedures (1) preferred coalition dynamic (PCD) and (2) preferred coalition greedy (PCG).
Note that $\mathbb{L}^{\theta, \eta}_j$ includes the subscribers from $\mathbb{W}^{\theta, \eta}_j$ and all $s_i \in \mathbb{P}^{\theta, \eta}_j$ such that $\delta_{i, j} > 0$. This pruning ensures that the non-SLA-complaint proposing subscribers are evicted and are not considered in the coalition formation (Lines 3-5 of Algorithm \ref{algo:PreferredCoalition}). Finally, the waiting list in the final iteration is the preferred coalition (Lines 15-18 Algorithm \ref{algo:Matching}).

\subsubsection{Working of PCD}
\begin{algorithm}[htb]
    \caption{\textbf{PreferredCoalition}}
    \label{algo:PreferredCoalition}
    \DontPrintSemicolon
    \small
    \KwIn{$\mathbb{W}^{\theta, \eta}_j, \mathbb{P}^{\theta, \eta}_j$}
    \KwOut{$\mu(c^{\theta, \eta}_{j})$}
    $\mu(c^{\theta, \eta}_{j}) \gets \emptyset$\;
    $\mathbb{L}^{\theta, \eta}_j \gets \mathbb{W}^{\theta, \eta}_j$\tcp{create a list of candidate EVs}
    \For{$s_i \in \mathbb{P}^{\theta, \eta}_j$}{
        \If{$\delta_{i, j}  > 0$}{\label{code:prunning}
            $\mathbb{L}^{\theta, \eta}_j \gets \mathbb{L}^{\theta, \eta}_j \cup \{s_j\}\,$\tcp{update the candidate list}
        }\vspace{-0.04in}
    }
    \tcc{Optimal preferred coalition search}
    $\mu(c^{\theta, \eta}_{j}) \gets \pcdp(\mathbb{L}^{\theta, \eta}_j)$\;
    \tcc{Greedy preferred coalition search}
    $\mu(c^{\theta, \eta}_{j}) \gets \pcgreedy(\mathbb{L}^{\theta, \eta}_j)$\;
    return $\mu(c^{\theta, \eta}_{j})$\;
\end{algorithm}
Algorithm~\ref{algo:PCD} first sorts the candidate EV list $\mathbb{L}^{\theta, \eta}_j$ by their \textit{minimum permissible charging rate}, defined as the ratio of each vehicle's charge requirement $\psi_{i, j}$ to its effective duration to SLA-compliant charging $\delta_{i,j}$, in non-increasing order, retained in $\hat{\mathbb{L}}^{\theta, \eta}_j$ (Line~\ref{code:PCD_Step1} of Algorithm~\ref{algo:PCD}). Such a consideration 
facilitates the selection of subscribers with higher charge requirements, increasing the total charge transferred and prioritizing subscribers' SLA.

From the sorted list $\hat{\mathbb{L}}^{\theta, \eta}_j$, the EVs that can be charged without breaching the SLA are identified. 

For memoization, two 2D matrices, $D_j$ and $A_j$, with dimensions $(|\mathbb{L}^{\theta, \eta}_j| + 1)$ rows and $(\delta_j^{max} + 1)$ columns, are considered. Here, the $k^{th}$ row represents the first $k$ EVs from $\hat{\mathbb{L}}^{\theta, \eta}_j$, and column $l$ represents the charging duration. 
Likewise, $D_j[k][l]$ denotes the maximum charge transferable in $l$ duration considering the first $k$ EVs, and $A_j[k][l]$ is the corresponding ordered assignment to achieve $D_j[k][l]$.
The $0^{th}$ row of $D_j$ is initialized to zero, and the $0^{th}$ row of $A_j$ is an empty list, reflecting that the maximum charge transfer is zero with no EVs available. Similarly, column $0$ indicates that with a duration of $0$, no EV can be charged (Line~\ref{code:zero row_start} -\ref{code:zero col_end} of Algorithm~\ref{algo:PCD}).
\begin{algorithm}[]
    \caption{\textbf{PCD}}
    \label{algo:PCD}
    \DontPrintSemicolon
    \small
  
    \KwIn{$\mathbb{L}^{\theta, \eta}_j$}
    \KwOut{$\mu(c^{\theta, \eta}_{j})$}
    
    $\mu(c^{\theta, \eta}_{j}) \gets \emptyset$\;
    $\hat{\mathbb{L}}^{\theta, \eta}_j \gets$ Sort $\mathbb{L}^{\theta, \eta}_j$ in descending order of $\psi_{i, j} / \delta_{i,j}$ \label{code:PCD_Step1}\;
    
    $\delta_j^{max}\gets \max(\delta_{i, j} : s_i \in \mathbb{L}^{\theta, \eta}_j)$\;
    $D_j \gets \left[0\right]_{|\mathbb{L}^{\theta, \eta}_j| + 1 \times \delta_j^{max}+ 1}$ \;
    Initialize empty 2-D matrix $A_j$ with $(|\mathbb{L}^{\theta, \eta}_j| + 1)$ rows and $(\delta_j^{max}+ 1)$ columns.\;

    \For{$l \gets 0$ to $\delta_j^{max}$}{ \label{code:zero row_start}
        $D_j[0][l] \gets 0$\;
        $A_j[0][l] \gets $ empty list\;
    }
    \For{$k \gets 0$ to $|\hat{\mathbb{L}}^{\theta, \eta}_j|$}{
        $D_j[k][0] \gets 0$\;
        $A_j[k][0] \gets $ empty list\label{code:zero col_end}\;
    }
    
    \tcc{Populate $D$}
    \For{$k \gets 1$ to $|\hat{\mathbb{L}}^{\theta, \eta}_j|$}{
        \For{$l \gets 1$ to $\delta_j^{max}$}{
            $i \gets \hat{\mathbb{L}}^{\theta, \eta}_j[k].getID()$\hfill \tcp{Get the subscriber ID}
            
            \If{$ \mathcal{T}_{i, j}  \leq l$ and $l \leq \delta_{i,j}$}{ \label{code: PCD_cond_include}
                \tcc{If $s_i$ is selected}
                \If{$\psi_{i, j} + D_j[k-1][l -  \mathcal{T}_{i, j}] \geq D_j[k-1][l]$ and $|s_i \cup A_j[k-1][l -  \mathcal{T}_{i, j} ]| \leq q^{\theta, \eta}_{j}$}{\label{code: PCD_cond_include2}
                    $D_j[k][l] \gets \psi_i^j + D_j[k-1][l - \mathcal{T}_{i, j} ]$\;
                    $A_j[k][l] \gets$ Append $s_i$ to $A_j[k-1][l -  \mathcal{T}_{i, j} ]$\;                         
                }
                \tcc{If $s_i$ is not selected}
                \Else{
                    $D_j[k][l] \gets D_j[k-1][l]$\;
                    $A_j[k][l] \gets A_j[k-1][l]$\;
                }
                
            }
            \Else{
                $D_j[k][l] \gets D_j[k-1][l]$\;
                $A_j[k][l] \gets A_j[k-1][l]$\;
            }
        }
    }
    \tcc{Select the maximum value in the last row of $D_j$}
    $l' \gets  \arg\max_{1 < l < \delta_j^{max}} (D_j[|\mathbb{L}^{\theta, \eta}_j|][l])$\;
    $\mu(c^{\theta, \eta}_{j}) \gets A_j[|\mathbb{L}^{\theta, \eta}_j|][l']$\;
    return $\mu(c^{\theta, \eta}_{j})$\;
\end{algorithm}
The row $k$ and column $l$ reflect a sub-problem involving the first $k$ EVs from $\hat{\mathbb{L}}^{\theta, \eta}_j$, given a maximum allowed charge duration of $l$. For the $k^{th}$ EV, identified as $s_i = \hat{\mathbb{L}}^{\theta, \eta}_j[k].getID()$, the procedure has the option to either include or exclude it from being assigned to $ c^{\theta, \eta}_{j}$. If excluded, the charge transferred in that case is captured in $D_j[k-1][l]$ as the problem simplifies to determining the maximum charge with the first $(k-1)$ EVs for a duration $l$.

On the other hand, an EV $s_i$ is included in the assignment \textit{iff} the following conditions ( Line~\ref{code: PCD_cond_include} ~\ref{code: PCD_cond_include2} of Algorithm~\ref{algo:PCD}) are satisfied: 
\begin{enumerate}
\item If $s_i$ can be charged within duration $l$ without violating the SLA (Line~\ref{code: PCD_cond_include} of Algorithm~\ref{algo:PCD});
\item The inclusion of $s_i$ does not exceed the capacity at CP $ c^{\theta, \eta}_{j}$;
\item The aggregate charge transfer including $s_i$ is greater than that of excluding $s_i$. Specifically, $\psi_{i, j} + D_j[k-1][l - \mathcal{T}_{i, j}] > D_j[k-1][l]$.
\end{enumerate} 

The total charge transfer, including $s_i$, considers the charge requirement of $s_i$ and the maximum charge that can be transferred in the remaining time after charging $s_i$. Whether $s_i$ is included or excluded, $A_j$ is updated accordingly. After populating matrices $D_j$ and $A_j$, the maximum SLA-compliant charge transfer is identified by finding the highest value in the last row of $D_j$. In case of multiple maximum values, the corresponding entries in $A_j$ are examined, and the one that charges the most EVs is selected to maximize both the total charge transfer and the number of EVs charged.
\subsubsection{Working of PCG}
Like PCD, the greedy coalition selection procedure is captured in Algorithm~\ref{algo:PCG}. 
The very first step in the procedure sorts the EV list $\mathbb{L}^{\theta, \eta}_j$ in non-increasing order considering their minimum permissible charging rates. It then sequentially processes $\mathbb{L}^{\theta, \eta}_j$ by sequentially adding it to the coalition subject to $c^{\theta, \eta}_{j}$, and SLA-adherence (Lines 5-11 of Algorithm \ref{algo:PCG}). Such a processing order prioritizes assignments of EVs seeking higher charges with relatively short waiting times, offering a locally optimal solution.
\begin{algorithm}[t]
    \caption{\textbf{PCG}}
    \label{algo:PCG}
    \DontPrintSemicolon
    \KwIn{$\mathbb{L}^{\theta, \eta}_j$}
    \KwOut{$\mu(c^{\theta, \eta}_{j})$}
    Initialize empty assignment sequence $A_j$\;
    $\mu(c^{\theta, \eta}_{j}) \gets \emptyset$\;
    $\mathcal{T}^e \gets 0$\;
    Sort $\mathbb{L}^{\theta, \eta}_j$ in descending order of their minimum permissible charging rate $\frac{\psi_{i, j}}{\delta_{i, j}}$\;
    \For{$s_i \in \mathbb{L}^{\theta, \eta}_j$}{
        \If{$|\mu(c^{\theta, \eta}_{j})| \geq q^{\theta, \eta}_{j}$}{
            break\;
        }
        \If{$\mathcal{T}^e + \mathcal{T}_{i, j}\leq \delta_{i, j}$}{
            $\mu(c^{\theta, \eta}_{j}) \gets \mu(c^{\theta, \eta}_{j}) \cup \{s_i\}$\;
            Append $s_i$ to $A_j$\;
            $\mathcal{T}^e \gets \mathcal{T}^e + \mathcal{T}_{i, j}$\;
        } 
    }
    $\mu(c^{\theta, \eta}_{j}) \gets A_j$\;
    return $\mu(c^{\theta, \eta}_{j})$\;
\end{algorithm}
\begin{figure*}[!htbp]
    \centering
    \begin{minipage}[t]{0.32\textwidth}
    \centering
    \includegraphics[width=\textwidth]{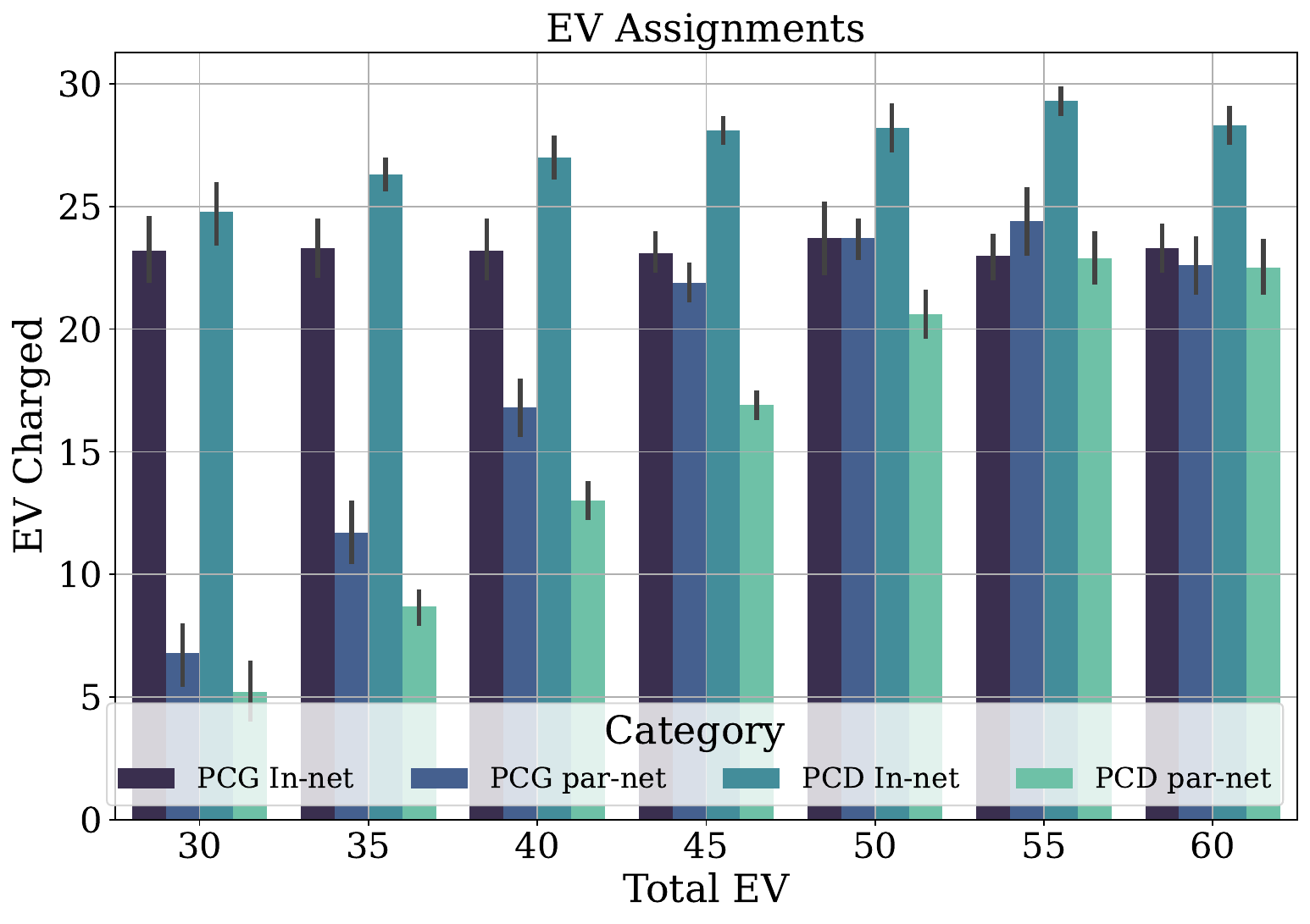}
    \vspace{-0.25in}
    \subcaption{EV Assignments. $q_j^{\theta, \eta} = 2$} 
    \label{fig:EV_Assignments}
    \end{minipage}
    \hfill
    \begin{minipage}[t]{0.32\textwidth}
    \centering
    \includegraphics[width=\textwidth]{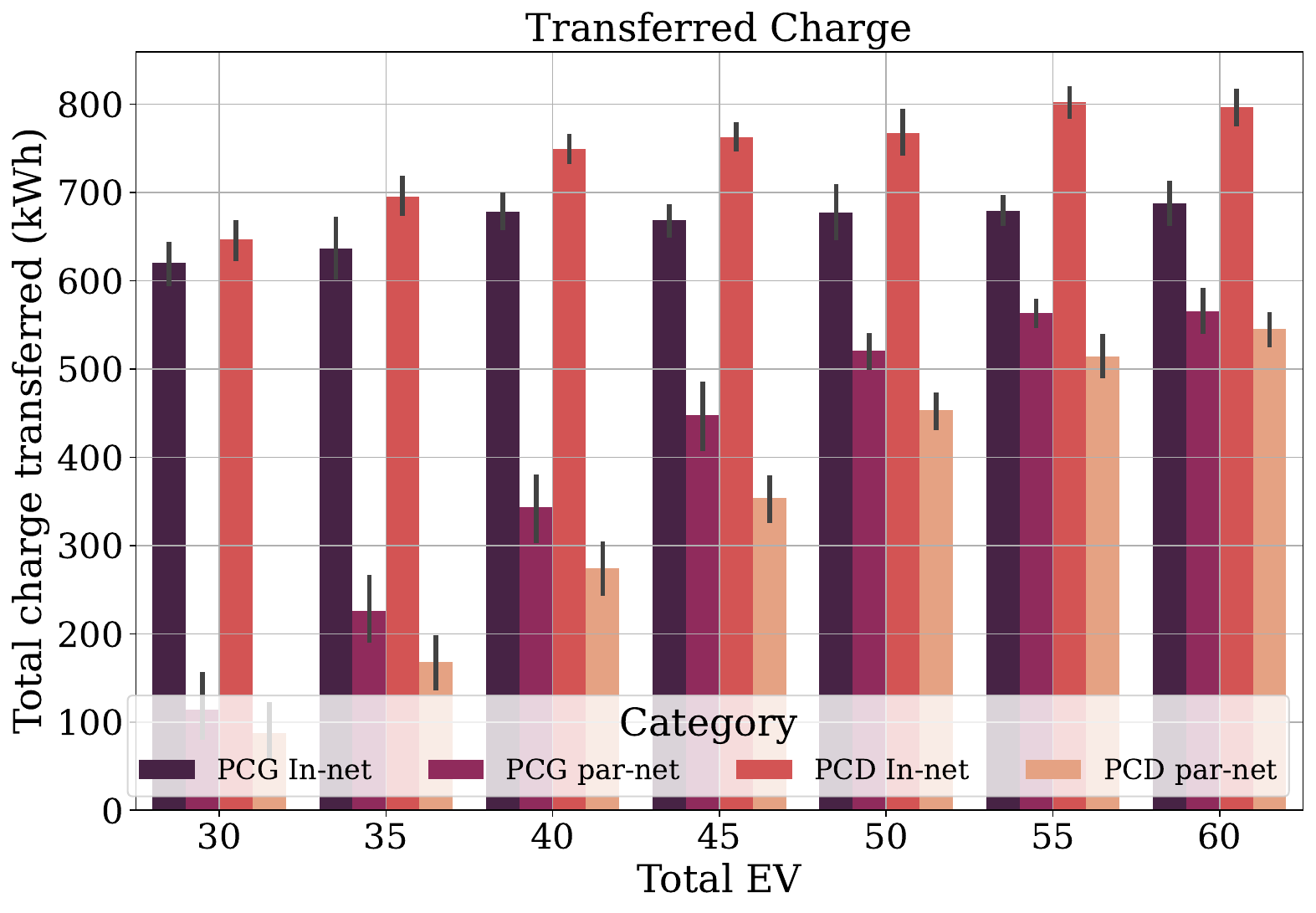}
    \vspace{-0.25in}
    \subcaption{Total Charge Transferred. $q_j^{\theta, \eta} = 2$} 
    \label{fig:Transferred_Charge}
    \end{minipage}
      \hfill
    \begin{minipage}[t]{0.348\textwidth}
    \centering
    \includegraphics[width=\textwidth]{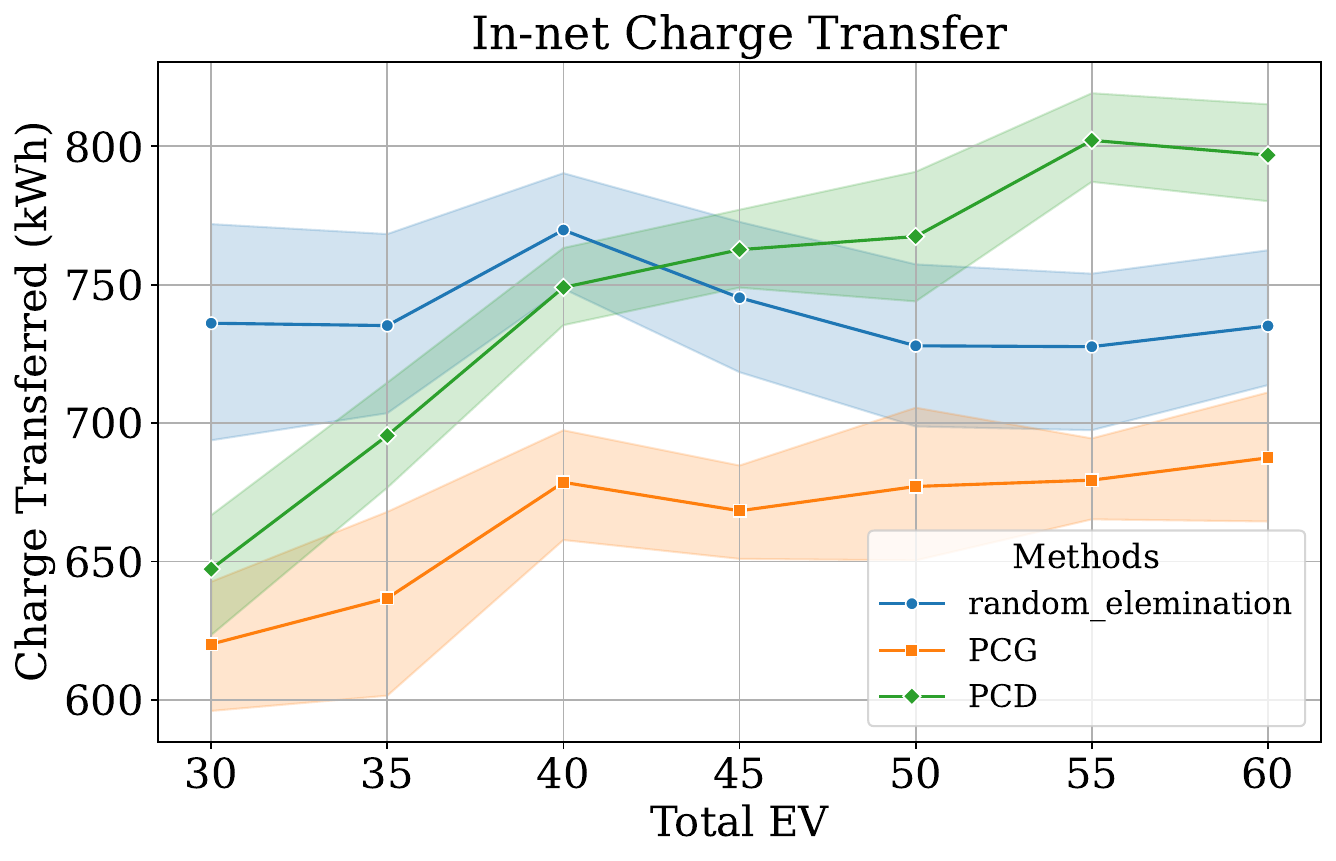}
    \vspace{-0.25in}
    \subcaption{In Network Charge Transferred. $q_j^{\theta, \eta} = 2$} 
    \label{fig:In-net_Charge_Transfer}
    \end{minipage}
    \caption{Performance based on Varying Number of EVs.}
\label{fig:performance_based_on_varying_EVs_part1}
\vspace{-0.2in}
\end{figure*}
\begin{figure}[!htb]
    \centering
\resizebox{\linewidth}{!}
    {
    \begin{minipage}[t]{0.65\linewidth}
    \centering
    \includegraphics[width=\columnwidth]{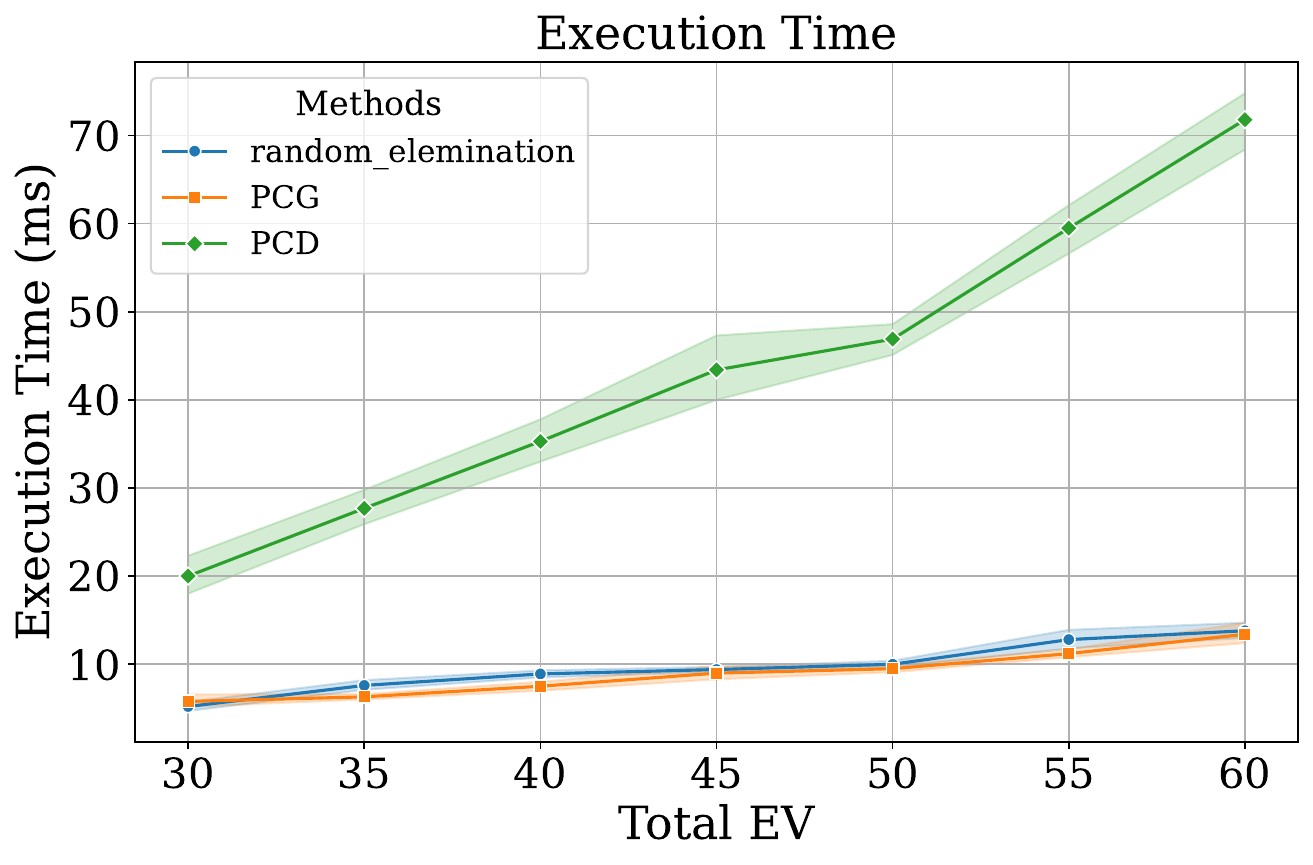}
    \subcaption{Execution Time (s).}
    \label{fig:execution_time}
    \end{minipage}
    \hfill
    \begin{minipage}[t]{0.65\linewidth}
    \centering
    \includegraphics[width=\textwidth]{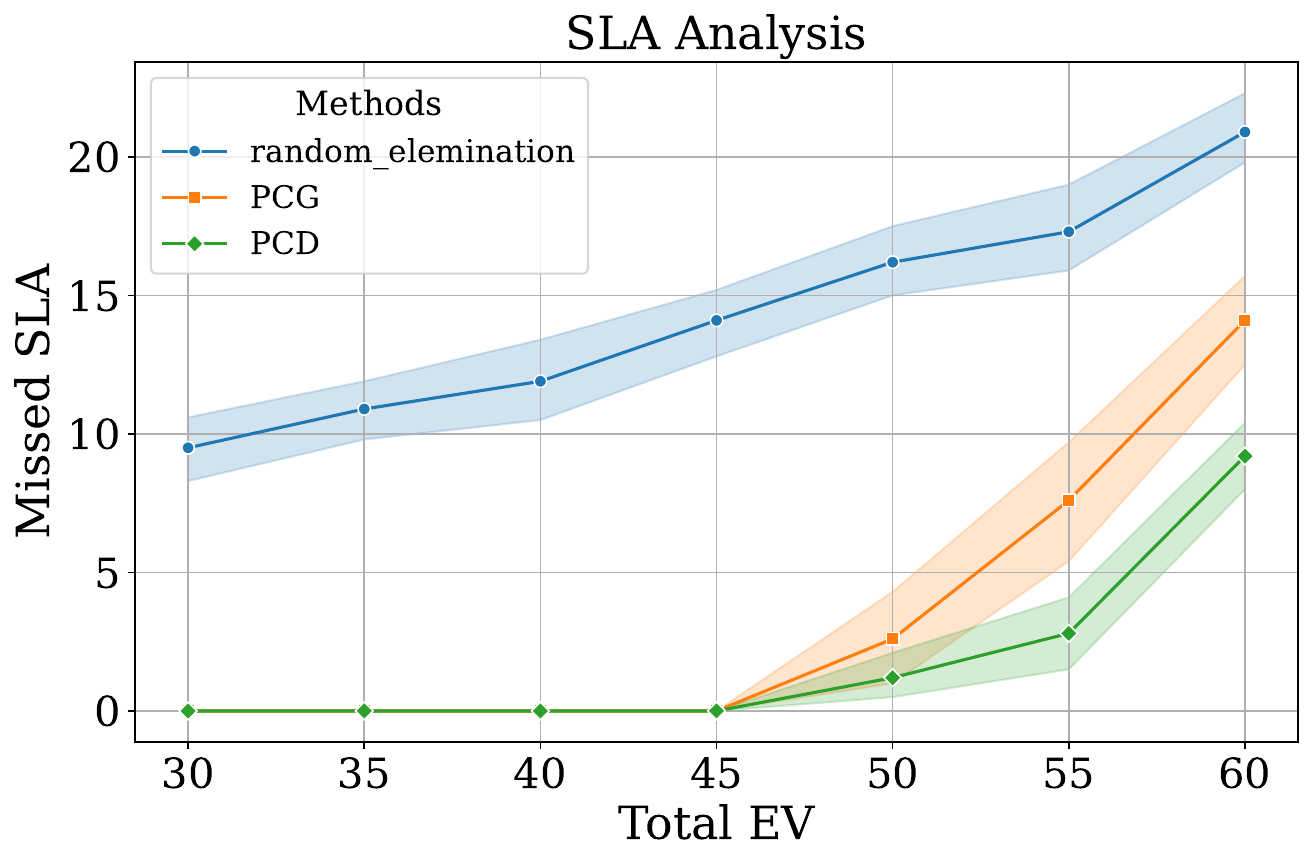}
    \subcaption{SLA Missed.} 
    \label{results:missed_SLA}
    \end{minipage}
    }
    \caption{Performance based on Varying Number of EVs.}
\label{fig:performance_based_on_varying_EVs_part2}
\end{figure}

\begin{theorem}
    The matching algorithm with \pcdp and \pcgreedy takes $\mathcal{O}(|\mathcal{C}|^2  |\mathcal{S}|  (  \log|\mathcal{C}| +  \delta_j^{max}))$ and $\mathcal{O}(|\mathcal{C}|^2  |\mathcal{S}|  \log|\mathcal{C}|)$ time respectively.
\end{theorem}
\vspace{-0.2in}
\begin{proof}
In the worst case, a subscriber may send proposals to all $|\mathcal{C}|$ CPs. With $|\mathcal{S}|$ subscribers in the system, the matching algorithm can be iterated up to $|\mathcal{C}| \cdot |\mathcal{S}|$ times. During each iteration, the preferred coalition algorithm is called for each CP. The preferred coalition algorithm mainly contains a \texttt{for} loop iterating over a proposer list of size $|\mathbb{P}^{\theta, \eta}_j|$ ($\leq |\mathcal{C}|$) and a function call to $\pcdp$ or $\pcgreedy$. The $\pcdp$ and $\pcgreedy$ algorithms both first sort the list of candidate EVs in $\mathcal{O}(|\mathcal{C}|  \log|\mathcal{C}|)$ time. The $\pcdp$ algorithm then fills 2D matrices $D_j$ and $A_j$ in $\mathcal{O}(|\mathcal{C}|  \delta_j^{max})$ time, while $\pcgreedy$ selects the first $q_j^{\theta, \eta}$ candidates that satisfy the SLA conditions in $\mathcal{O}(|\mathcal{C}|)$ time. Thus, the matching algorithm with $\pcdp$ takes $\mathcal{O}(|\mathcal{C}|^2  |\mathcal{S}|  (  \log|\mathcal{C}| + \delta_j^{max}))$ time, and with $\pcgreedy$, it takes $\mathcal{O}(|\mathcal{C}|^2  |\mathcal{S}| \log|\mathcal{C}|)$ time. 
\end{proof}
\vspace{-0.1in}
Generally, the maximum effective duration of SLA-compliant charging, $\delta_j^{max}$ exceeds the value $\log |\mathcal{C}|$.
This implies that matching EVs to CPs using the $\pcdp$ algorithm generally requires more time than when using the $\pcgreedy$ algorithm.
\vspace{-0.1in}
\section{Experimental Results}\label{sec:results}
This section highlights the environmental setup and discusses key insights from the experimental observations.
 \vspace{-0.1in}
\subsection{Environmental Setup}
The proposed methodology has been implemented in \texttt{Python 3.8.3}, with EVs and charging points modeled as separate classes and parameter specifications in Table \ref{tab:EnvParams}. We simulate an RSU region as a $16 \times 16$ Chicago-style 2-D grid where each block is a square region bordered by four edges, symbolizing road segments. Within this RSU region, the coordinates of EVs needing charging are uniformly generated at random.
Algorithms~\ref{algo:pref_generator_EV} and \ref{algo:PrefHelper} are implemented as methods within the EV class, while Algorithms~\ref{algo:Matching} through \ref{algo:PCG} are implemented as methods within the RSU. Our implementation is publicly available at~\cite{EV_CP_Assignment2024}.
\begin{table}[!htb]
\caption{Key Parameter Values}
\small
\label{tab:EnvParams}
\centering
\begin{tabular}{|l|l|}
\hline
\textbf{Parameter}               & \textbf{Value/Range}                               \\ \hline
City block length (Chicago grid) & $\frac{1}{8}$ mile                                 \\ \hline
RSU radius                       & $1.5$ mile ($\equiv 12$ block length) \\ \hline
Area under RSU                   & $16 \times 16$ blocks                              \\ \hline
$|\mathcal{C}^{1, in}|$          & 5                                                  \\ \hline
$|\mathcal{C}^{0, in}|$          & 10                                                 \\ \hline
$|\mathcal{C}^{1, par}|$         & 5                                                  \\ \hline
$|\mathcal{C}^{0, par}|$         & 10                                                 \\ \hline
$\hat{r}_i$                      & $2$ kWh \\ \hline
$v_i$                      & $30$ mph \\ \hline
$\alpha_i$                      & $0.8$ \\ \hline
$B_{full}$                      & $60$ kWh \\ \hline
$B_i^0$                      & $U[10,37]$ kWh \\ \hline
$c_i^{1, *}$                      & $U[0,60]$ kWh \\ \hline
$\gamma_i$                      & $U[5,25]$ min \\ \hline
$m_i$                      & $U[3,4]$ mpW \\ \hline
$r_j$ (regular Charging)                     & $1$ kWh per min\\ \hline
$r_j$ (fast Charging)                     & $2$ kWh per min\\ \hline
\end{tabular}%
\end{table}
\vspace{-0.1in}
\subsection{Results \& Analysis}
Figs. \ref{fig:performance_based_on_varying_EVs_part1} and \ref{fig:performance_based_on_varying_EVs_part2} capture the behavior of the different schemes for varying EVs. Fig. \ref{fig:EV_Assignments} captures the EVs successfully charged for fixed capacity at fast and regular CPs within the vendor's in-network and partner network. It shows that the in-network CPs become exhausted as the number of EVs increases, forcing them to charge at partner networks in both approaches. Moreover, the fast CPs are exhausted earlier due to pre-defined preference ordering (refer to Section \ref{sec:subscription_model}). In this case, PCD outperforms the PCG in assigning more EVs, as it eliminates the least preferred EVs using exhaustive search rather than locally optimal decisions made in the latter.

Fig. \ref{fig:Transferred_Charge} captures the total charge transferred as the number of EVs increases in the RSU region. It can be inferred from the figure that PCD satisfies the goal of achieving higher in-network charge transfer. The inferior performance of PCG is attributed to the heuristic-based selection criterion, wherein the EVs are processed in decreasing order of the minimum permissible charging rate. This result demonstrates that in-network CPs prioritize EVs with more significant charge requirements and lower SLAs. Thus, it favors increased charge-consuming and swift customer base for the vendor as per the goal of our work. Other sluggish EVs with meager charging needs are pushed to partner network to meet their SLAs. This can be affirmed in Fig. \ref{fig:Transferred_Charge}, as we observe more EVs assigned to partner network charging stations due to the bounded capacity of the in-network charging, leading to SLA violations. 

\begin{figure*}[htbp]
    \centering
    \begin{minipage}[t]{0.32\textwidth}
    \centering
    \includegraphics[width=\textwidth]{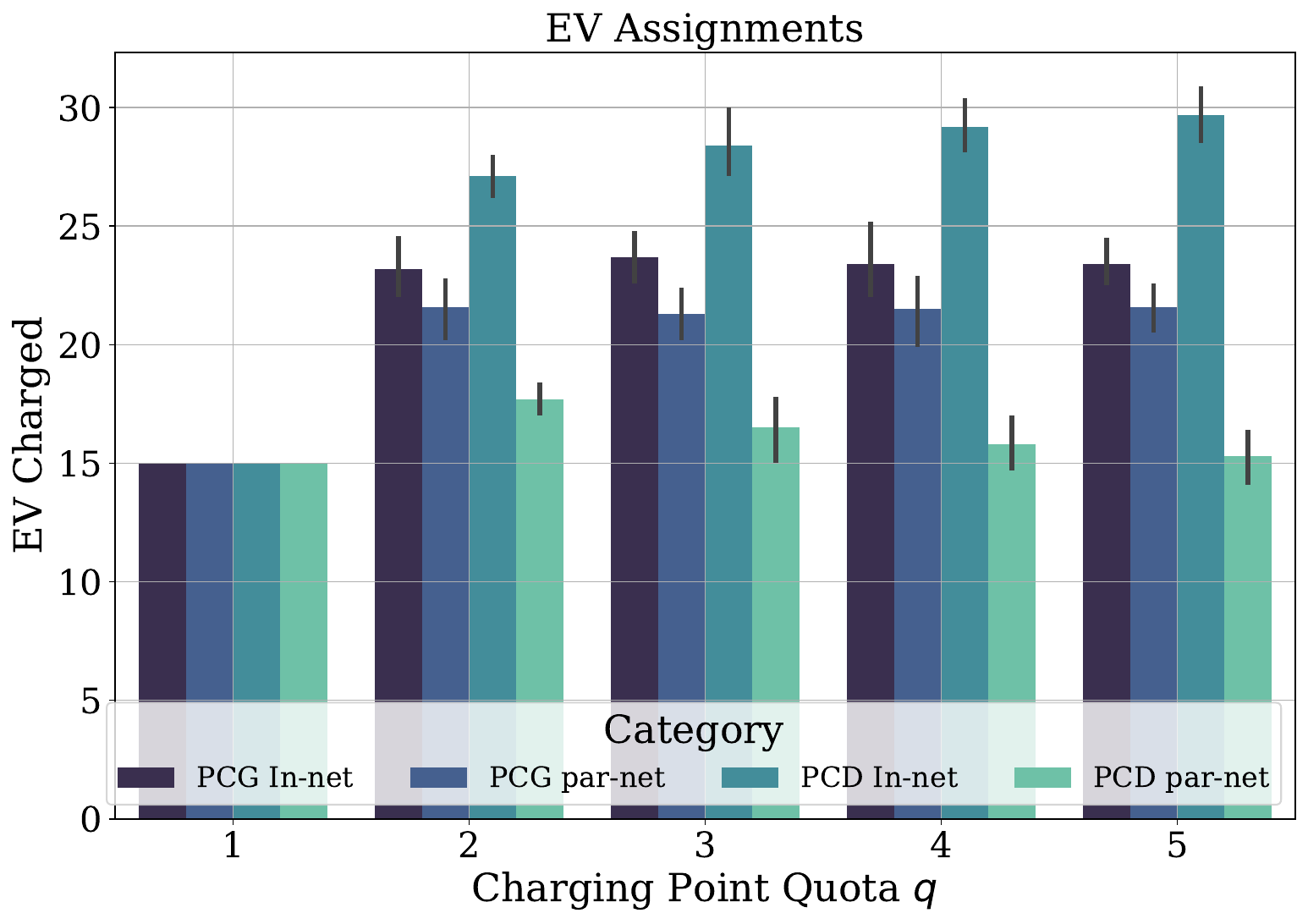}
    \vspace{-0.15in}
    \subcaption{EV Assignments. $|\mathcal{S}| = 45$} 
    \label{fig:EV_Assignments_varyingQ}
    \end{minipage}
    \hfill
    \begin{minipage}[t]{0.32\textwidth}
    \centering
    \includegraphics[width=\textwidth]{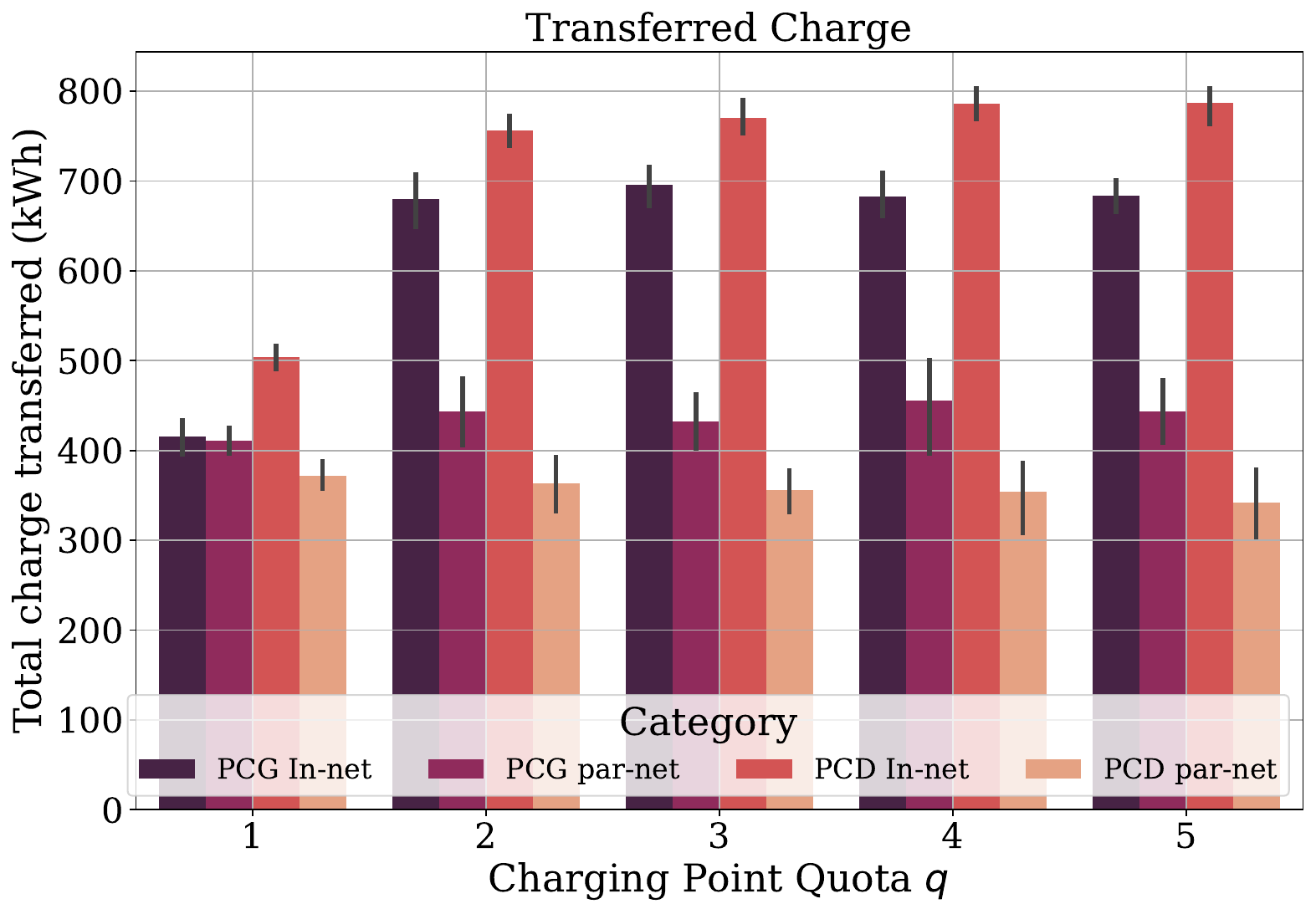}
    \vspace{-0.15in}
    \subcaption{Total Charge Transferred. $|\mathcal{S}| = 45$} 
    \label{fig:Transferred_Charge_varyingQ}
    \end{minipage}
      \hfill
    \begin{minipage}[t]{0.35\textwidth}
    \centering
    \includegraphics[width=\textwidth]{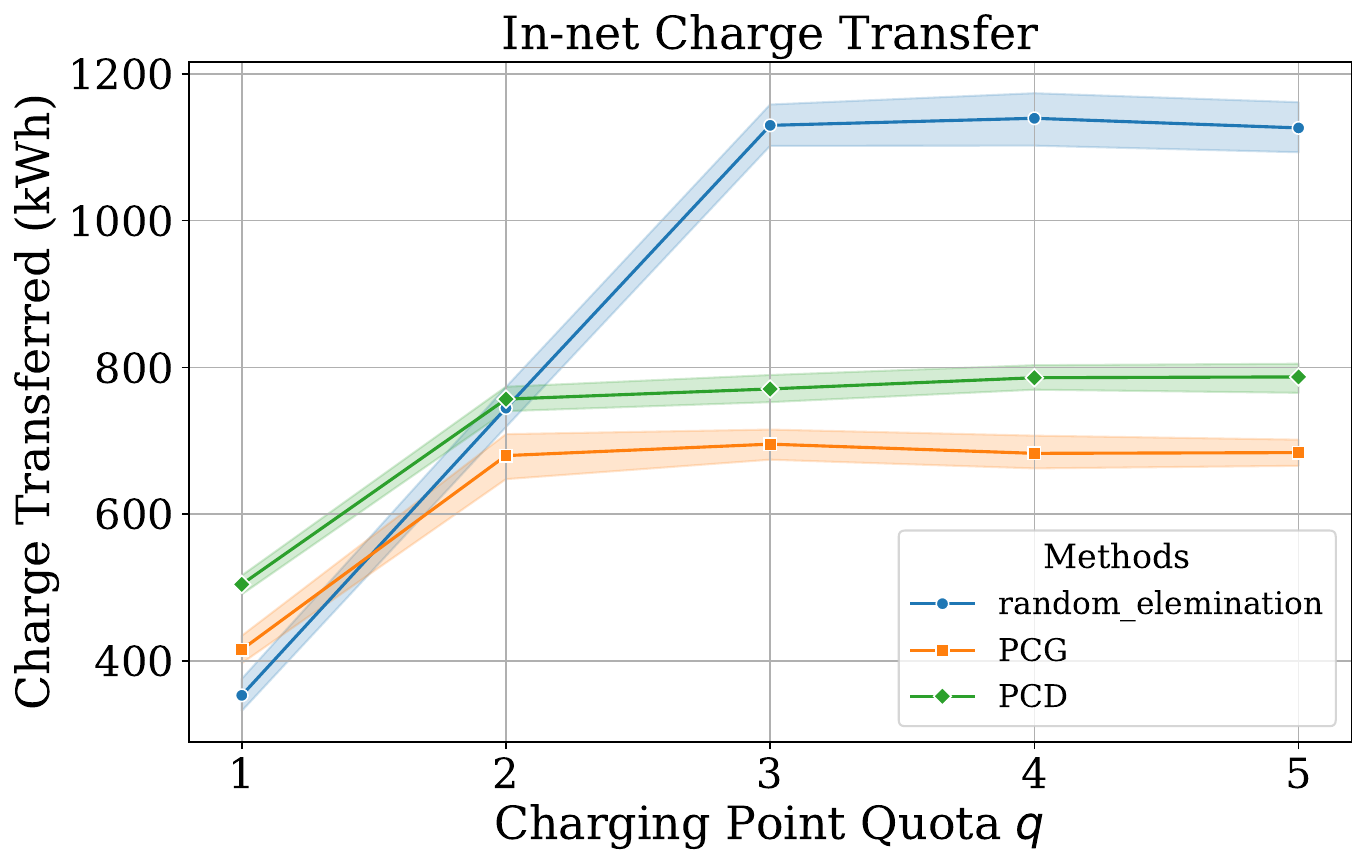}
    \vspace{-0.15in}
    \subcaption{In Network Charge Transferred. $|\mathcal{S}| = 45$} 
    \label{fig:In-net_Charge_Transfer_varyingQ}
    \end{minipage}
    \vspace{-0.05in}
    \caption{Performance based on Varying Queue Size.}
    \label{fig:Performance_varyingQ}
    \vspace{-0.15in}
\end{figure*}

Fig. \ref{fig:In-net_Charge_Transfer} captures the total in-network charge transferred for PCG, PCD, and a random elimination technique considering the CP capacity without SLA. As discussed earlier, the in-network charge transfer in PCD is higher than in PCG due to its exhaustive search-based eviction policy. An exciting pattern of more significant in-network charge transfer via $random\_elimination$ is observed for some test cases. This occurs as it fills the CPs to their maximum capacity, leading to the undesirable situation of a higher number of EVs missing their SLAs, as observed in Fig. \ref{results:missed_SLA}. As expected, the execution overhead in Fig. \ref{fig:execution_time} shows that PCD executes longer than $random$\_$elimination$ and PCG due to exhaustive searches. 

On the other hand, Figs. \ref{fig:Performance_varyingQ} and \ref{fig:queue_size_analysis} capture the performance of the different schemes for varying CP capacity, bounded by the total EVs within the RSU region as 45. Increased CP capacity accommodates more EVs as expected in Fig. \ref{fig:EV_Assignments_varyingQ}. 
As with the case of varying number of EVs in Figs. \ref{fig:EV_Assignments} and \ref{fig:Transferred_Charge}, varying the CP capacity also supports high in-network charge transfer with PCD and PCG along with more EVs being assigned, as depicted in Figs. \ref{fig:EV_Assignments} and \ref{fig:Transferred_Charge}. 
Note that when the queue size is $1$, all the in- and partner charging points operate at maximum capacity. As the queue size increases, the in-network charging options are utilized more due to enforced preference ordering. Alternatively, Fig. \ref{fig:In-net_Charge_Transfer_varyingQ} captures the comparative behavior of the in-network charge transferred considering PCG, PCD, and random$\_$elimination procedures. We observe an exciting behavior wherein the PCG and PCD have consistent in-network charges transferred irrespective of CP capacity ($\geq 2$). This is because, at this point, no more vehicles can be queued up at in-network charging points without missing their respective SLAs and, hence, are pushed to partner networks. The behavior and the reasoning for random$\_$elimination follow the discussion provided for Fig. \ref{fig:In-net_Charge_Transfer}. 
\begin{figure}[h]
    \centering
    \vspace{-0.12in}
\resizebox{\linewidth}{!}
    {
    
    \begin{minipage}[t]{0.65\linewidth}
    \centering
    \includegraphics[width=\columnwidth]{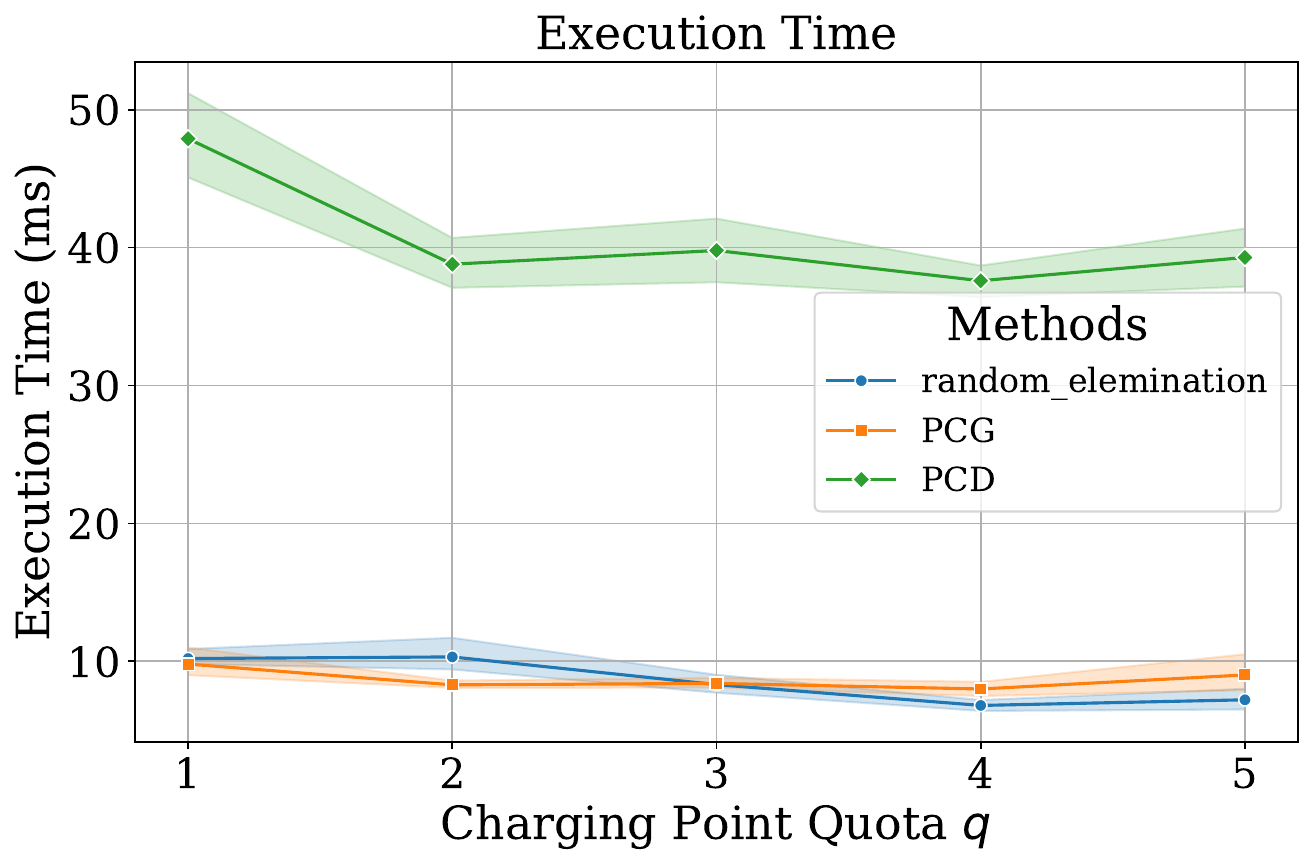}
    \subcaption{Execution Time (s).}
    \label{fig:execution_time_varyingQ_plot1}
    \end{minipage}
    \hspace{-0.05in}
    \begin{minipage}[t]{0.65\linewidth}
    \centering
    \includegraphics[width=\textwidth]{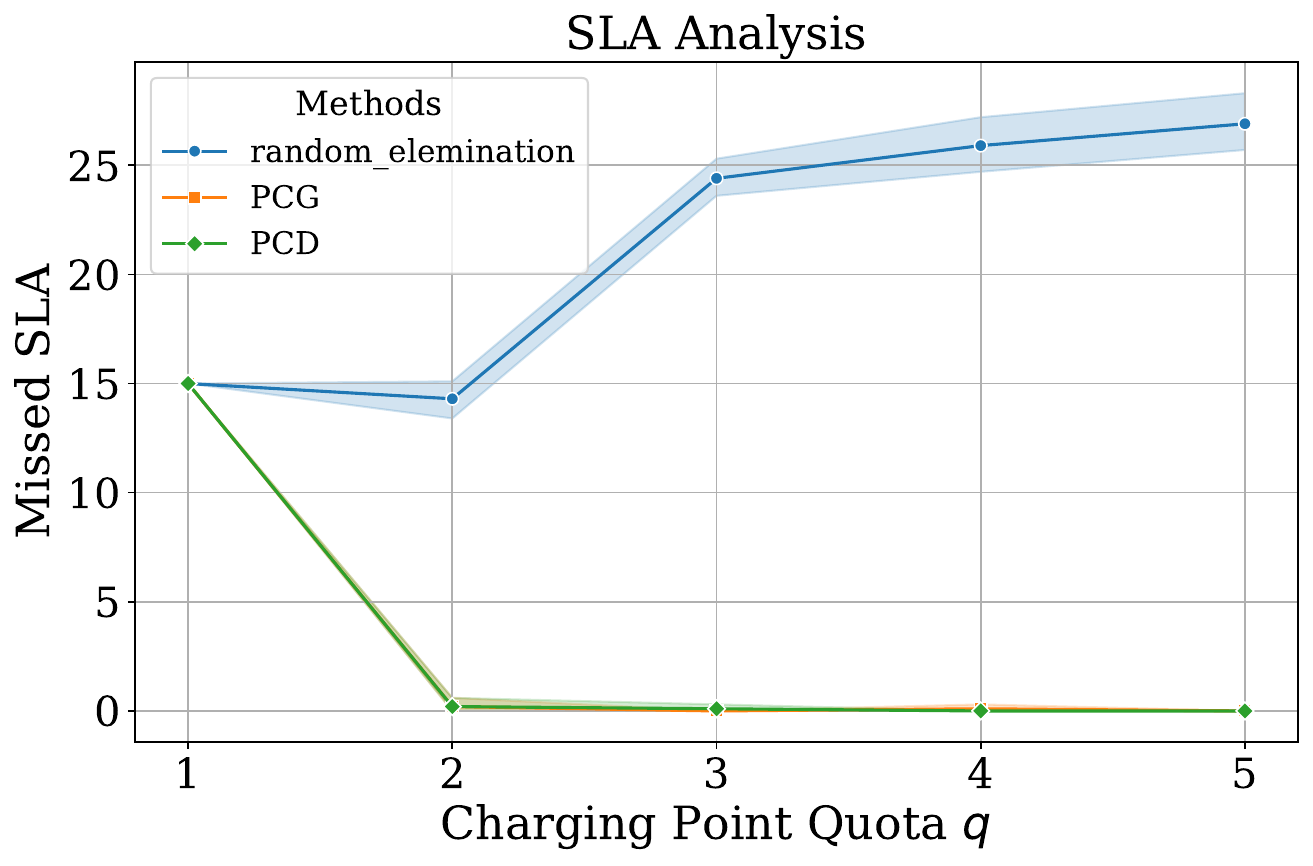}
    \subcaption{SLA Missed.} 
    \label{results:missed_SLA_varyingQ_plot2}
    \end{minipage}
    }
    \caption{Performance based on Varying Queue Size.}
    \label{fig:queue_size_analysis}
    \vspace{-0.15in}
\end{figure}
Figs. \ref{fig:execution_time_varyingQ_plot1} and \ref{results:missed_SLA_varyingQ_plot2} illustrate the execution times and SLA misses. The execution time of PCD is higher than that of PCG and random elimination due to the exhaustive PCD search procedure. However, we observe a decreasing trend with increasing queue size for all the approaches. This is because for a fixed number of vehicles, as the queue size increases, the number of times PCD and PCG are invoked is reduced owing to faster assignments, as the EVs have more options. On the other hand, the number of SLA misses is observed to be a sharp decline for PCD and PCG as the queue size increases to 2, as the EVs can be queued up at charging points. Moreover, both approaches experience 100\% SLA adherence ($\geq 2$). Finally, Figs \ref{fig:EV_distribution_n_ev45_q4} and \ref{results:Charge_Distribution_n_ev45_q4} respectively highlight the number of EVs assigned and the charge distributed across different charging points. The assignments follow the preference order, wherein the in-network charging points are filled first, followed by the partner network. The higher assignments in the figure indicate that the fast charging points are preferred within the in and partner network.
\begin{figure*}[h]
\centering
    \begin{minipage}[t]{0.49\linewidth}
    \centering
    \includegraphics[width=\columnwidth]{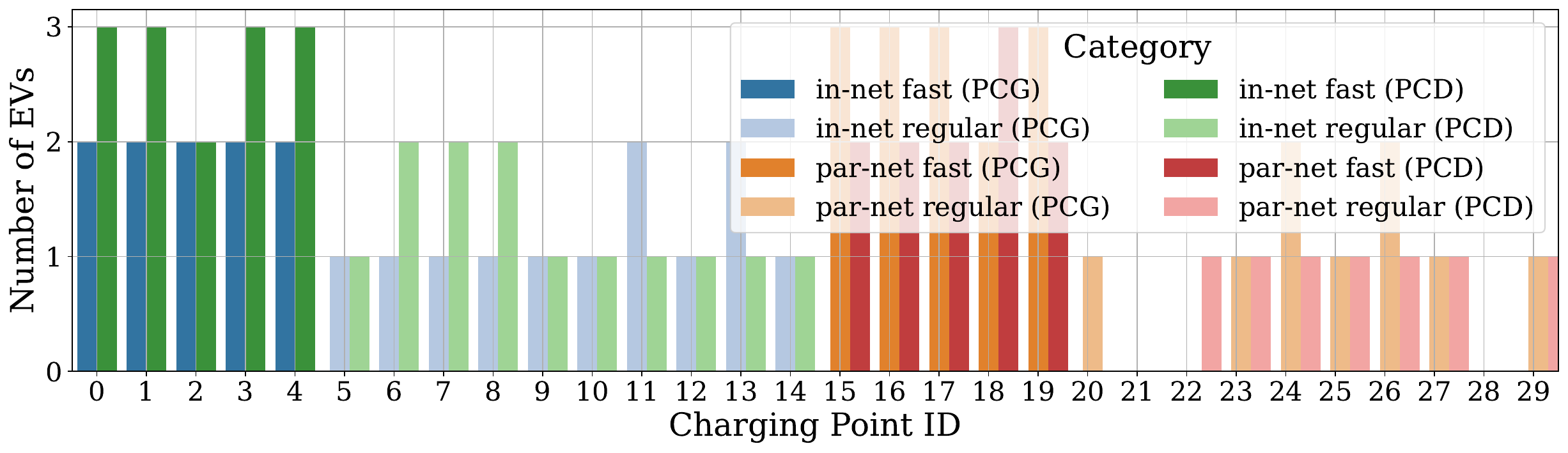}
    \vspace{-0.25in}
    \subcaption{EV distribution.} 
    \label{fig:EV_distribution_n_ev45_q4}
    \end{minipage}
    \hspace{-0.05in}
    \begin{minipage}[t]{0.49\linewidth}
    \centering
    \includegraphics[width=\columnwidth]{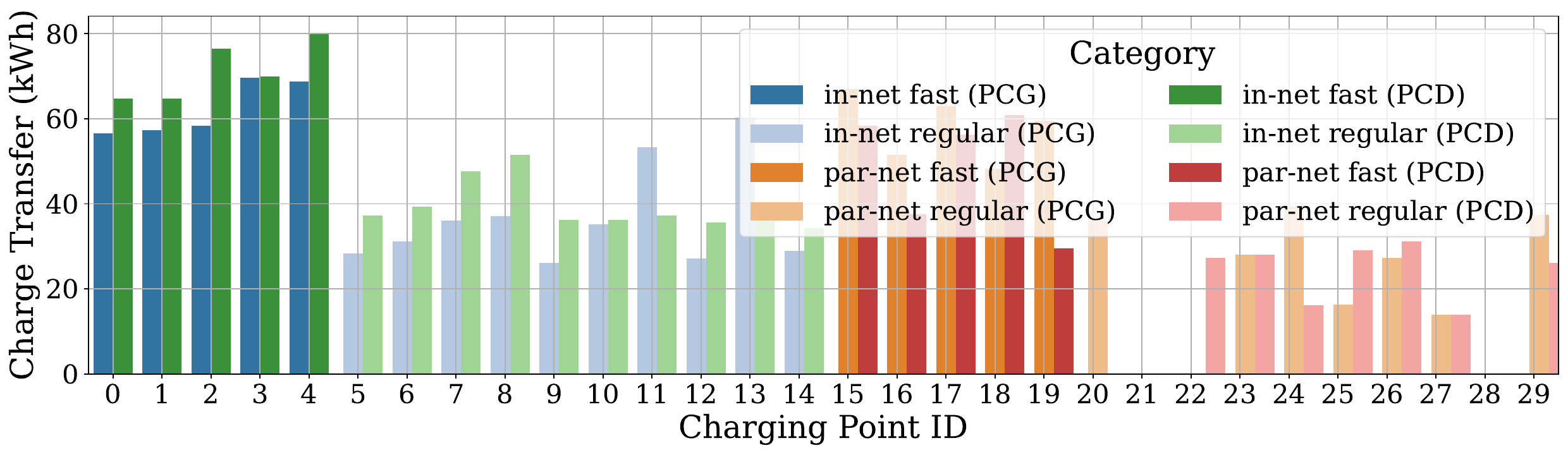}
    \vspace{-0.25 in}
    \subcaption{Charge Distribution.} 
    \label{results:Charge_Distribution_n_ev45_q4}
    \end{minipage}
    \vspace{-0.05in}
    \caption{Distribution across the charging points. $|\mathcal{S}| = 45$, $q_j^{\theta, \eta} = 4$.}
    \vspace{-0.2in}
\end{figure*}

\noindent \textbf{Results Summary:} It is observed that 64.5\% and 70.6\% of the average charge transfers are in-network with PCG and PCD for increasing EVs, respectively. Also, 58.7\% and 65\% of the average charge transfers are in-network with PCG and PCD for varying CP capacity, respectively. The average execution time of PCD is higher by 5.64$\times$ and 5.85$\times$, the average execution time of PCG and random elimination, respectively. For in-network charging, PCG and PCD give a corresponding gain of 14.6\% and 20.8\% over random elimination. PCG achieves 89\% of the optimal coalition produced by PCD. On average, random elimination has a massive 37.7\% of unserved EVs within RSU due to being unassigned to a CP or missed SLA, while PCG and PCD abandon 3\% and 0.1\% of EVs, respectively.
\vspace{-0.2in}
\section{Features and Limitations}\label{sec:discussion_and_limitations}
{\our} offers several key features. Firstly, it assumes that EVs are subscribed to a vendor plan with predefined SLAs that limit the maximum wait time before charging, with higher-tier plans providing shorter wait times. The plan currently sets a maximum quota for fast charging but does not restrict regular charging. Secondly, considering these SLAs, {\our} generates a stable, fast, and efficient vendor-centric EV-CP assignment, prioritizing the in-network assignment of EVs with fast charging quotas to fast charging stations first, followed by regular charging. Though the proposed scheme shows promising performance, it has the following limitations.

\begin{itemize}[leftmargin=*]
    \item \textbf{A Comprehensive Subscription Model}: 
    The current subscription model is straightforward and offers lower wait times (SLA) to EVs subscribed to higher-tier plans. However, determining optimal pricing for different tiers is complex and influenced by factors such as the location and type of CSs, where urban areas and faster charging options typically cost more. Factors such as time of use, energy source, subscription plans, local utility rates, and additional taxes or fees contribute to the overall cost. In {\our}, the fast charging limit is a fixed number. However, determining this quota is non-trivial and depends on the fast chargers available, electricity pricing, and other factors.
    \item \textbf{Stochastic Arrival and Departures of Vehicles}: Currently, {\our} operates on a static snapshot of the system, processing charging requests in batches and not accounting for dynamic arrivals or departures, which presents a more complex scenario.
    \item \textbf{Traffic Considerations en-route}: Once the RSU assigns a charging point, the time taken by the EV to reach the location can be significantly affected by traffic conditions, leading to chaotic situations if the EVs arrive either earlier or much later than expected.
    \item \textbf{Consideration of User-Preferred  Charging Time Slots}: 
    The current proposal does not consider the preferred time slots for EV charging.
    A more realistic approach would involve incorporating user preferences for specific charging time slots, which would add complexity to the scheduling problem.   
\end{itemize}
\vspace{-0.15in}
\section{Conclusions}\label{sec:cnls}
This work introduced the {\our} framework to create an EV-CP assignment plan accommodating fast and regular charging options. The challenge is inherently complex due to the rapid growth of EV adoption, the slow expansion of charging infrastructure, and the unpredictable nature of charging demands. The contributions of this work are two-fold: (1) an SLA-driven subscription model that limits the maximum waiting time for an EV at an assigned charging point and (2) a matching-theoretic solution that is stable, scalable, and effective in promptly assigning charge-seeking EVs to charging points. The matching procedure determines the preferred coalition using two strategies: (1) a greedy, heuristic-driven approach that is time-efficient and locally optimal, and (2) a dynamic programming approach that, while time-consuming, yields an optimal coalition. Comparative studies confirm that {\our} achieves performance improvements regarding in-network charge transferred, SLA compliance, and execution times.

\begin{acks}
This work was supported by the NSF grants: "Cyberinfrastructure for Accelerating Innovation in Network Dynamics" (CANDY) under award \# OAC-2104078, and  ``Satisfaction and
Risk-aware Dynamic Resource Orchestration in
Public Safety Systems" (SOTERIA) under award \# ECCS-2319995.
\end{acks}

\bibliographystyle{ACM-Reference-Format}
\bibliography{finalsub_main}

\end{document}